\documentclass[
reprint,
 amsmath,amssymb,
 aps,
floatfix
]{revtex4-2}
\usepackage{mathtools}
\usepackage{graphicx}%
\usepackage{dcolumn}%
\usepackage{bm}%
\usepackage{braket}
\usepackage{tikz}
\usepackage{multirow}
\usepackage{enumitem}
\usepackage{quantikz}

\usepackage{amsmath, amssymb}
\usepackage{algpseudocode}
\usepackage[linesnumbered,ruled,vlined,procnumbered]{algorithm2e}
\usepackage{setspace}
\usepackage[version=4]{mhchem}

\def \eqand {\text{ and }}

\newtheorem{theorem}{Theorem}
\newtheorem{lemma}{Lemma}

\begin{document}

\preprint{APS/123-QED}

\title{Geo-ADAPT-VQE: Quantum Information Metric-Aware Circuit Optimization for Quantum Chemistry }%

\author{Mohammad Aamir Sohail}
\email{mdaamir@umich.edu}
\altaffiliation{This research was conducted at MERL.}
\affiliation{%
Electrical and Computer Engineering, University of Michigan, Ann Arbor, USA}
\author{Toshiaki Koike-Akino}
\email{koike@merl.com}
\affiliation{%
Mitsubishi Electric Research Laboratories (MERL), Cambridge, MA 02139, USA}

\begin{abstract}
Adaptive ansatz construction has emerged as a powerful technique for reducing circuit depth and improving optimization efficiency in variational quantum eigensolvers. However, existing adaptive methods, including ADAPT-VQE, rely solely on first-order gradients and therefore ignore the underlying geometry of the quantum state space, limiting both convergence behavior and operator-selection efficiency. We introduce Geo-ADAPT-VQE, a geometry-aware adaptive VQE algorithm that selects operators from a pool using the natural gradient rule. The geometric operator-selection rule enables the ansatz to grow along directions aligned with the underlying quantum-state geometry, thereby improving convergence and reducing the algorithm’s susceptibility to shallow local minima and saddle-point regions. We further provide an asymptotic convergence result. We present numerical simulations involving five molecules 
\ce{H5}, \ce{LiH}, \ce{HF}, \ce{H2O}, and \ce{BeH2}, which
 demonstrate that Geo-ADAPT-VQE achieves faster and more stable convergence compared to existing methods, while producing significantly shorter ansatz. In particular, Geo-ADAPT achieves up to $100$x reduction in energy error compared to existing methods.
\end{abstract}

\maketitle

\section{\label{sec:level1}Introduction}

Quantum chemistry is recognized as an area where quantum computing can offer practical computational advantages \cite{cao2019quantum}. This is primarily due to the rapid growth of the underlying quantum state space: the Hilbert space associated with molecular systems scales exponentially with the number of orbitals \cite{helgaker2013molecular,bader1991quantum,woolley1976quantum}. As a consequence, classical algorithms become intractable for computing the electronic structure of many strongly correlated systems, even with advanced approximations. These limitations have motivated the pursuit of quantum algorithms for solving electronic structure problems, which several studies suggest could provide a decisive advantage ~\cite{cao2019quantum,aspuru2005simulated,mcardle2020quantum,lanyon2010towards,whitfield2011simulation,reiher2017elucidating,bauer2020quantum}.
The Quantum Phase Estimation Algorithm (QPE)~\cite{kitaev1995quantum} was the first major algorithmic development for simulating electronic structure problems using quantum computers, which offers an exponential speedup over classical algorithms \cite{abrams1997simulation,abrams1999quantum}. 
However, QPE requires large circuit depths to achieve accurate phase resolution. This necessity leads to long coherent evolution times needed to implement these circuits, which exceed the capabilities of current noisy intermediate-scale quantum (NISQ) devices ~\cite{preskill2018quantum}. Even with algorithmic improvements that help reduce some overhead~\cite{cruz2020optimizing,berry2019qubitization}, practical demonstrations of QPE have, thus far, been limited to only a small number of molecular systems~\cite{o2019quantum,tranter2025high,paesani2017experimental,o2016scalable}.

To address these limitations, the Variational Quantum Eigensolver (VQE)~\cite{peruzzo2014variational} was introduced as a hybrid quantum–classical alternative designed for obtaining ground-state energies of molecular and many-body Hamiltonians on NISQ devices. In VQE, a parameterized quantum circuit, known as an ansatz, prepares a trial wavefunction, and a classical optimizer iteratively updates the circuit parameters to minimize the variational energy \cite{mcclean2016theory}. VQE has since been demonstrated on a variety of hardware platforms, for example, superconducting qubits ~\cite{colless2018computation,kandala2017hardware} and trapped ions \cite{hempel2018quantum,yung2014transistor,shen2017quantum}.
A central challenge in VQE is selecting an ansatz that is expressive enough to accurately estimate the ground state while still remaining shallow enough for NISQ hardware and manageable for classical optimization. One of the most widely used families of ansatz in quantum chemistry is the unitary coupled cluster singles and doubles (UCCSD) ansatz \cite{bartlett2007coupled,taube2006new}. UCCSD is the unitary form of the classical coupled-cluster method, adapted for implementation as a quantum circuit. In this ansatz, the wavefunction is generated by applying exponentials of single- and double-excitation operators, each associated with its own variational parameter. 

Despite its strong foundation in quantum chemistry, UCCSD presents several practical challenges:
\textit{1. Large Circuit Depth}. Implementing UCCSD ansatz on hardware requires mapping each single and double excitation into a sequence of one- and two-qubit gates. As a result, as the system size increases, incorporating all of these excitation operators results in quantum circuits that are significantly deeper \cite{romero2018strategies}. \textit{2. Large number of parameters}. Each excitation introduces its own variational parameter. As the pool of singles and doubles expands, the resulting high-dimensional optimization landscape becomes difficult to navigate. This complexity also makes the landscape more sensitive to noise, leading to challenges such as barren plateaus \cite{ryabinkin2018qubit,mcclean2018barren,cerezo2021cost}.
\textit{3. Trotterization ambiguity}.
To implement UCCSD on hardware, the full coupled-cluster operator must be factorized into a product of exponentials. However, different operator orderings yield distinct approximate states because Trotterization is not exact when only a finite number of steps is used. Consequently, UCCSD is not uniquely defined on a quantum device, and its performance varies depending on the specific factorization chosen \cite{evangelista2019exact}.
In practice, fixed-form ansatz such as UCCSD often contain many excitation operators that contribute little or no improvement to the energy, yet still inflate the circuit depth and the number of variational parameters \cite{fedorov2022unitary}. 

These limitations naturally prompt the question: can we construct the ansatz more selectively, adding only those operators that meaningfully reduce the energy?
This perspective motivates the class of adaptive ansatz methods, in which the circuit is constructed iteratively rather than predetermined. A notable example is ADAPT-VQE \cite{grimsley2019adaptive}, which introduced the concept of dynamically selecting operators based on gradient information. ADAPT-VQE consists of an outer loop that incrementally builds the ansatz and an inner optimization subroutine that updates its parameters. In the outer loop, ADAPT-VQE selects the most promising operator from a predefined pool based on its contribution to the energy gradient and adds it to the current ansatz. In the inner loop, once the ansatz is extended, a classical optimizer adjusts all variational parameters to minimize the energy before the next operator-selection step. This approach leads to shorter ansatz with fewer parameters.
Subsequent developments have further strengthened the adaptive‐ansatz paradigm introduced by ADAPT-VQE \cite{tang2021qubit,anastasiou2024tetris,bespalova2025k,stadelmann2025strategies,vaquero2025pruned,ramoa2025reducing}. For example, qubit-ADAPT-VQE \cite{tang2021qubit} extends the operator pool from fermionic excitations to qubit-space operators, enabling more hardware-efficient constructions and offering improved performance in settings where fermionic mappings introduce additional overhead. More recently, the TETRIS-ADAPT-VQE \cite{anastasiou2024tetris} showed that inserting mutually commuting operators in structured layers can significantly reduce circuit depth by allowing for the simultaneous execution of gates.

Despite recent advancements, it is important to note that adaptive VQE algorithms fundamentally rely on first-order gradient information for operator selection and do not account for the underlying geometry of the quantum state space. This presents a significant limitation, as geometry-aware methods in quantum learning and variational quantum algorithms have been shown to greatly enhance convergence. For instance, quantum natural gradient descent (QNGD) \cite{stokes2020quantum} that updates parameters along the steepest-descent direction given by a quantum information metric, such as the Fubini-Study metric in the space of pure quantum states. Recent works have shown that QNGD provides an advantage in optimizing parameterized quantum systems by taking optimization paths more aligned with the underlying geometric structure of quantum states, compared to other strategies \cite{tao2023laws,minervini2025quantum,wierichs2020avoiding,gacon2021simultaneous,yamamoto2019natural,sohail2025quantum,atif2022quantum,wang2023experimental,koczor2022quantum}. Motivated by these observations, we propose the Adaptive Geometric-Assembled Problem-Tailored Variational Quantum Eigensolver (Geo-ADAPT-VQE), a geometry-aware adaptive variational quantum eigensolver that employs a geometric operator-selection rule along with a metric-based inner optimization subroutine.

We summarize the key contributions of this work.

 \noindent\textit{1. Geometry-Aware Adaptive Ansatz Construction.} We develop a geometry-aware adaptive variational quantum eigensolver, named Geo-ADAPT-VQE (see Fig.~\ref{fig:geoADAPT} and Algorithm\ref{alg:geo-adapt-vqe}). The method incorporates the intrinsic geometry of the quantum state space via the Fubini–Study metric \cite{petz1996geometries,petz1998information} into the operator selection rule. At each iteration, among all operators in the pool, the one corresponding to the steepest descent direction in the space of quantum states is selected.

\noindent \textit{2. Convergence of Geo-ADAPT-VQE.} 
We establish asymptotic convergence guarantees for Geo-ADAPT-VQE (see Theorem \ref{thm:asymptotic_conv}). We prove that the Geo-ADAPT energy sequence converges and that the pool natural gradient vanishes asymptotically, implying that no operator in the pool can further decrease the energy. Using the quadratic geometric information (QGI) inequality~\cite{sohail2025quantum}, we further show that Geo-ADAPT converges globally to the minimum energy achievable by the operator pool. Furthermore, this convergence occurs at an exponential rate, with a diminishing residual error. To our knowledge, this provides one of the first convergence analyses for adaptive VQE methods.

\noindent \textit{3. Position-Optimized Geo-ADAPT.} 
We introduce Pos-Geo-ADAPT, an extension of Geo-ADAPT-VQE that optimizes not only which operator to add using a geometry-aware operator selection rule but also where to insert it within the ansatz. While the Geo-ADAPT and the ADAPT-VQE algorithms append operators at the end of the circuit, the Pos-Geo-ADAPT allows insertion at the beginning, between existing operators, or at the end. This positional flexibility enhances expressibility and reduces the energy error $(\mathrm{E} - \mathrm{E}_{\mathrm{FCI}})$, where $\mathrm{E}$ denotes the energy obtained by the algorithm and $\mathrm{E}_{\mathrm{FCI}}$ is the full configuration interaction (FCI) ground-state energy \cite{helgaker2013molecular}.

\noindent \textit{4. Numerical Evaluation.} 
We provide extensive numerical benchmarks of Geo-ADAPT and Pos-Geo-ADAPT against ADAPT-VQE, VQE with gradient descent (GD), and QNGD. We consider five molecular systems, \ce{H5}, \ce{LiH}, HF, \ce{H2O}, and \ce{BeH2}, across multiple bond lengths. The results demonstrate that Geo-ADAPT achieves more than an $80\%$ reduction in the effective ansatz complexity (EAC) required to reach chemical accuracy compared to standard VQE. Moreover, it exhibits up to $2$x faster convergence and requires up to $4$x fewer parameters than ADAPT-VQE to reach chemical accuracy. Furthermore, Geo-ADAPT achieves up to a $100$x reduction in energy error compared to other methods.

In what follows, we begin by reviewing the necessary background on VQE, the UCCSD ansatz, and ADAPT-VQE in Section~\ref{sec:background}. Our main contributions appear in Section~\ref{sec:main_results}: the Geo-ADAPT-VQE algorithm is introduced in Section~\ref{subsec:geoadapt}, its convergence properties are established in Section~\ref{subsec:convergence}, and numerical experiments on molecular systems are presented in Section~\ref{subsec:numerical}.

\section{Background}
\label{sec:background}

\noindent\textit{\textbf{VQE for Quantum Chemistry.}}
VQE is a hybrid quantum--classical algorithm designed to approximate the ground-state energy of molecular systems. In quantum chemistry, the electronic structure problem aims to determine the ground-state energy of the electronic Hamiltonian, which in second quantization can be written as
\begin{equation}
    \hat{H} = \sum_{pq} h_{pq} a_p^\dagger a_q + \frac{1}{2} \sum_{pqrs} h_{pqrs} a_p^\dagger a_q^\dagger a_r a_s,
\end{equation}
where \(a_p^\dagger\) and \(a_p\) are fermionic creation and annihilation operators acting on spin orbitals. The indices \(p,q,r,s\) label spin orbitals in the chosen basis set, typically derived from a Hartree--Fock (HF) calculation that partitions the orbitals into occupied and virtual subsets. The coefficients \(h_{pq}\) and \(h_{pqrs}\) correspond to the one- and two-electron integrals, which encode the kinetic energy, electron--nuclear attraction, and electron--electron Coulomb interactions. These integrals are computed from the molecular orbitals obtained in a chosen atomic basis (e.g., STO-3G, cc-pVDZ).

To simulate this Hamiltonian on a quantum computer, the fermionic operators are mapped onto qubit operators using transformations such as the Jordan--Wigner\cite{fradkin1989jordan}, Bravyi--Kitaev \cite{bravyi2002fermionic,tranter2015b}, or parity mapping \cite{yordanov2020efficient}. This produces a qubit Hamiltonian of the form \(\hat{H} = \sum_i c_i P_i\), where \(P_i\) are tensor products of Pauli operators (Pauli strings) acting on \(n\) qubits and \(c_i \in \mathbb{R}\) are the corresponding coefficients. The number of qubits required corresponds to the number of spin orbitals included in the active space.
VQE employs a parameterized quantum circuit (ansatz) \( |\Psi(\boldsymbol{\theta})\rangle = U(\boldsymbol{\theta})|0\rangle \) to prepare trial wavefunctions. The energy of the system, which is represented by the expectation value of the Hamiltonian corresponding to the ansatz $|\Psi(\boldsymbol{\theta})\rangle$ is given by:
\begin{equation}
    E(\boldsymbol{\theta}) = \langle \Psi(\boldsymbol{\theta}) | \hat{H} | \Psi(\boldsymbol{\theta}) \rangle.
\end{equation}
This energy value is computed on a quantum processor by measuring the expectation values of the individual Pauli terms \(P_i\). A classical optimizer iteratively updates the circuit parameters \(\boldsymbol{\theta}\) to minimize \(E(\boldsymbol{\theta})\). By the variational principle, the minimum value \(E(\boldsymbol{\theta}^*)\) provides an upper bound to the true ground-state energy \(E_0\), i.e., \(E(\boldsymbol{\theta}^*) \ge E_0\). In practical molecular simulations, the goal is to reach \emph{chemical accuracy}, corresponding to an energy error of less than \(1~\text{kcal/mol} \approx 1.6\times10^{-3}\) Hartree.

\vspace{5pt}
\noindent{\textbf{UCCSD Ansatz.}}
\label{subsec:UCCSD Ansatz}
A widely used chemically motivated ansatz is the Unitary Coupled Cluster with Singles and Doubles (UCCSD). The variational wavefunction is constructed by applying a unitary operator to the Hartree--Fock reference state \(|\Psi_{\mathrm{HF}}\rangle\):
\[
|\psi_{\mathrm{UCCSD}}\rangle = \widehat{U}_{\mathrm{UCCSD}} |\Psi_{\mathrm{HF}}\rangle,
\qquad
\widehat{U}_{\mathrm{UCCSD}} = e^{\,\mathrm{i}(T - T^\dagger)}.
\]
The cluster operator \(T\) is separated into single and double excitation components,
$T = T_1 + T_2,$
where
\[
T_1 = \sum_{i,a} \theta_i^{a} a_a^\dagger a_i, \qquad
T_2 = \sum_{i,j,a,b} \theta_{ij}^{ab} a_a^\dagger a_b^\dagger a_j a_i.
\]
Here, \(i,j\) label occupied orbitals and \(a,b\) label virtual orbitals in the Hartree--Fock reference configuration. The amplitudes \(\theta_i^{a}\) and \(\theta_{ij}^{ab}\) are taken to be real-valued so that \(i(T - T^\dagger)\) is Hermitian, ensuring that \(\widehat{U}_{\mathrm{UCCSD}}\) is unitary.

Since the excitation operators do not commute, the exact exponential cannot be implemented directly on current hardware. A first-order Trotter expansion approximates the unitary as a product of exponentials of individual excitation generators:
\[
\widehat{U}_{\mathrm{UCCSD}}^{(1)} \approx 
\bigg( \prod_{\mu\in\{ia\}} e^{\,\mathrm{i} \theta_{\mu} \hat{\tau}_{\mu}} \bigg)
\bigg( \prod_{\nu\in\{ijab\}} e^{\,\mathrm{i} \theta_{\nu} \hat{\tau}_{\nu}} \bigg),
\]
where each \(\hat{\tau}_{ia}\) and \(\hat{\tau}_{ijab}\) is an anti-Hermitian excitation operator defined as
\begin{align*}
    \hat{\tau}_{ia} &= (a_a^\dagger a_i - a_i^\dagger a_a) \eqand
\hat{\tau}_{ijab} =  (a_a^\dagger a_b^\dagger a_j a_i - a_i^\dagger a_j^\dagger a_b a_a).
\end{align*}
This representation explicitly separates the single- and double-excitation unitaries, with each term corresponding to a physically interpretable excitation process. The first-order Trotterized UCCSD wavefunction is therefore expressed as
\[
|\Psi_{\mathrm{UCCSD}}^{(1)}\rangle =
\widehat{U}_{\mathrm{UCCSD}}^{(1)}
|\Psi_{\mathrm{HF}}\rangle,
\]
which clearly shows the ansatz as a sequential application of parameterized unitary rotations generated by single and double excitation operators.

\begin{figure*}[]
    \centering
    \includegraphics[scale=0.31]{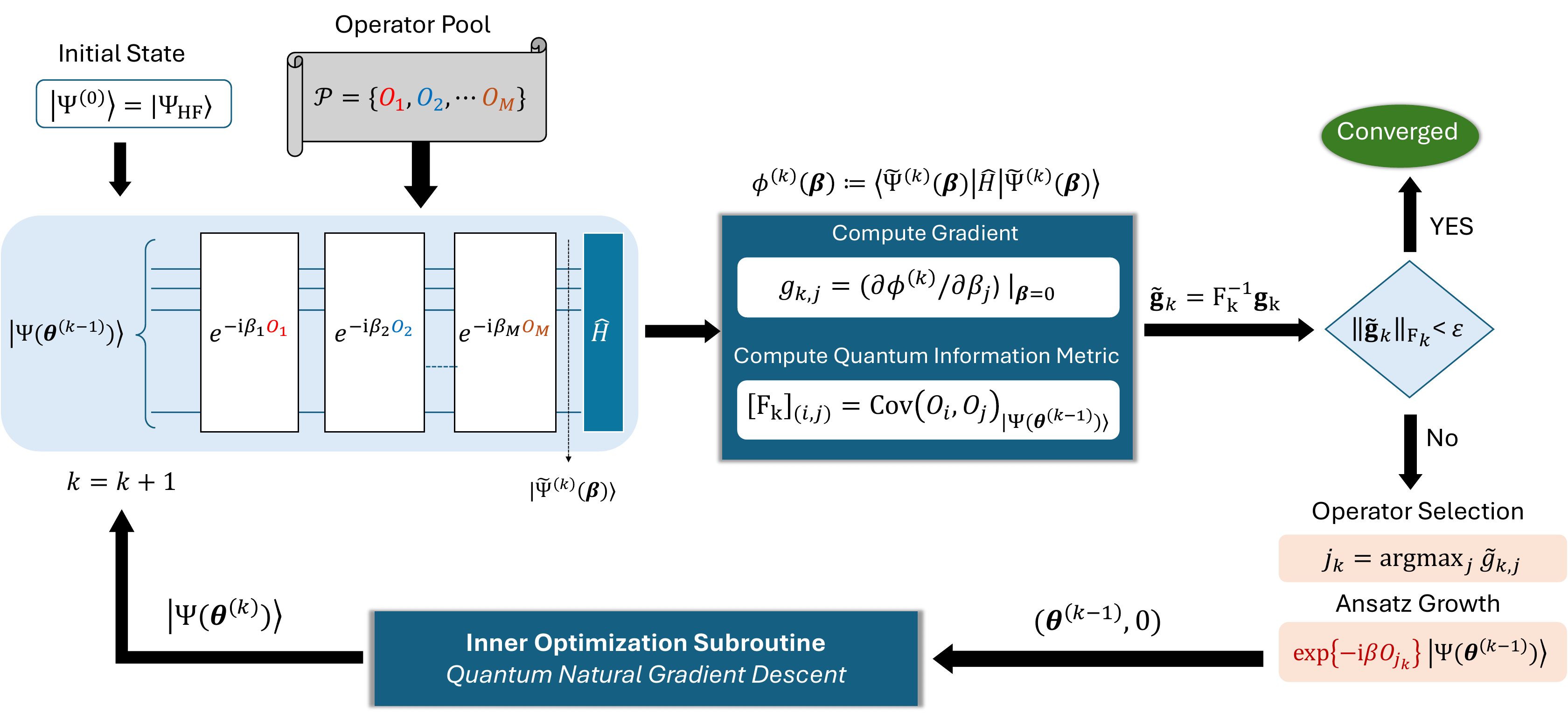}
\caption{
Overview of the proposed Geo-ADAPT algorithm. 
Starting from the Hartree–Fock reference state $|\Psi_{\mathrm{HF}}\rangle$, the method 
iteratively constructs an ansatz from an operator pool $\mathcal{P}$. 
At the $k$-th outer iteration, the algorithm evaluates the gradient 
$\mathbf{g}_{k}$ and the corresponding pool information metric $\mathrm{F}_{k}$ to form 
the pool natural gradient $\tilde{\mathbf{g}}_{k}$. 
A geometric operator-selection rule chooses the next operator 
$j_{k} = \arg\max_j |\tilde{g}_{k,j}|$, and the ansatz is extended by 
$\exp(-\mathrm{i} \beta O_{j_{k}})$. 
The resulting parameters are optimized through an inner loop using 
QNGD, producing the updated ansatz $|\Psi(\boldsymbol{\theta})^{(k)}\rangle$. 
The procedure repeats until the norm of the natural gradient 
$\|\tilde{\mathbf{g}}_{k}\|_{\mathrm{F}_{k}}$ falls below a prescribed tolerance, 
indicating convergence.
}\label{fig:geoADAPT}
\end{figure*}

\vspace{5pt}
\noindent{\textbf{ADAPT-VQE.}} The Adaptive Derivative-Assembled Pseudo-Trotter Variational Quantum Eigensolver (ADAPT-VQE)~\cite{grimsley2019adaptive} is an adaptive variational algorithm that constructs the ansatz iteratively rather than fixing it a priori. At each iteration, operators are selected from a predefined pool based on their energy gradients with respect to the current state. The main steps of the ADAPT-VQE procedure are summarized as follows:
\begin{itemize}[leftmargin=10pt]

    \item \textit{State Initialization:} Initialize the ansatz with the Hartree–Fock state.
    \item \textit{Operator Selection:} Evaluate the energy gradients for all operators in the pool with respect to the current state. Select the operator with the largest gradient magnitude and append it to the ansatz. If the gradient norm falls below a chosen threshold, terminate the algorithm.
    \item \textit{Parameter Optimization:} Re-optimize all variational parameters of the updated ansatz to minimize the energy. Repeat the operator-selection and optimization steps until convergence.
\end{itemize}

\section{Main Results}
\label{sec:main_results}
\noindent\textbf{Geo--ADAPT--VQE.}\label{subsec:geoadapt}
It adaptively constructs the variational ansatz by utilizing the geometric structure of the space of quantum states via a quantum information (Riemannian) metric tensor. The algorithm selects an operator from the predefined pool according to a quantum information-metric-based criterion. This criterion generalizes the gradient-based selection rule of ADAPT-VQE to the space of quantum states by selecting an operator that yields the steepest energy descent with respect to the quantum information geometry.

We now describe the Geo-ADAPT algorithm, and summarize it in Algorithm \ref{alg:geo-adapt-vqe} and illustrated in Fig.~\ref{fig:geoADAPT}. 

\noindent \textit{1. Operator Pool.} Define an operator pool \[\mathcal{P} := \{O_1, O_2, \cdots, O_M\}\] that includes all possible excitation generators considered for inclusion in the ansatz. For instance, \(O_j\) may represent generators of single and double excitations, such as \(\hat{\tau}_{ia}\) and \(\hat{\tau}_{ijab}\), respectively, as outlined in Section \ref{subsec:UCCSD Ansatz}.

\vspace{3pt}
\noindent \textit{2. Classical Preprocessing.}
After constructing the operator pool $\mathcal{P}$, the algorithm includes classical preprocessing steps analogous to standard VQE. On the classical computer, the one- and two-electron integrals of the molecular Hamiltonian are computed using a chosen basis set.
Each fermionic operator, including both pool operators and Hamiltonian terms, is then transformed into its qubit representation via a fermion-to-qubit mapping, such as the Jordan--Wigner or Bravyi--Kitaev transformation \cite{fradkin1989jordan,bravyi2002fermionic,tranter2015b,yordanov2020efficient}.
These preprocessed Hamiltonian and excitation operators form the inputs to the adaptive procedure executed on the quantum hardware.

\vspace{3pt}
\noindent \textit{3. State Initialization.}
The algorithm begins with a reference state, such as the Hartree–Fock ground state \(|\Psi_{\mathrm{HF}}\rangle\). At this point, no excitations are included in the ansatz.

\vspace{3pt}
\noindent
Next, suppose at the end of $(k\!-\!1)$-th iteration, the ansatz is given as
\[
\ket{\Psi(\boldsymbol{\theta}^{(k-1)})}
:= \bigg(\prod_{t=1}^{k-1} e^{-\mathrm{i}\theta_{j_t}^{(k-1)} O_{j_t}}\bigg)\ket{\Psi_{\mathrm{HF}}},
\]
where
\(
\boldsymbol{\theta}^{(k-1)} := (\theta_{j_1}^{(k-1)}, \theta_{j_2}^{(k-1)}, \ldots, \theta_{j_{k-1}}^{(k-1)})
\)
contains all optimized parameters up to iteration $(k\!-\!1)$.
Here, $j_t$ is the index selected at iteration $t$, $O_{j_t} \in \mathcal{P}$ denotes the corresponding operator, and $\theta_{j_t}$ is the optimized parameter associated with it. The corresponding energy is given as
\[
E(\boldsymbol{\theta}^{(k-1)})
= \bra{\Psi(\boldsymbol{\theta}^{(k-1)})}\hat{H}\ket{\Psi(\boldsymbol{\theta}^{(k-1)})}.
\]
For convenience, we use the following shorthand notation in the rest of the paper
\[
|\Psi^{(k-1)}\rangle := |\Psi(\boldsymbol{\theta}^{(k-1)})\rangle\eqand
E^{(k-1)} := E(\boldsymbol{\theta}^{(k-1)}).
\]
Now, at the beginning of the $k$-th iteration, the previously optimized parameters
$\boldsymbol{\theta}^{(k-1)}$ remain fixed while selecting the next operator to append.
For the operator pool $\mathcal{P}$, we define a \emph{trial-extended} ansatz as
\[
|\tilde{\Psi}^{(k)}(\boldsymbol{\beta})\rangle
= U_{\mathcal{P}}(\boldsymbol{\beta})\,|\Psi^{(k-1)}\rangle,
\]
where $U_{\mathcal{P}}(\boldsymbol{\beta}):=\prod_{j=1}^{M} e^{-\mathrm{i}\beta_{j} O_{j}}$ is the product of the unitary generate by the operators in $\mathcal{P}$ and $\boldsymbol{\beta}:= (\beta_1, \beta_2, \ldots, \beta_M)$ denotes the parameters associated with the extended ansatz.
The corresponding energy is given as
\begin{align*}
    \phi^{(k)}(\boldsymbol{\beta})
&= \langle \tilde{\Psi}^{(k)}(\boldsymbol{\beta}) |\, \hat{H} \,| \tilde{\Psi}^{(k)}(\boldsymbol{\beta}) \rangle.
\end{align*}
Note that by definition $\phi^{(k)}(\boldsymbol{0})=E(\boldsymbol{\theta}^{(k-1)})$.
The partial derivative of $\phi^{(k)}$ at $\boldsymbol{\beta}=\boldsymbol{0}$ is given as
\[
g_{k,j}
:= \partial_j \phi^{(k)}(\boldsymbol{\beta})|_{\boldsymbol{\beta}=\boldsymbol{0}} = -\mathrm{i}\,\langle \Psi^{(k-1)} | [\hat{H}, O_j] | \Psi^{(k-1)} \rangle .
\]
Collecting these components, we get the gradient vector:\[
\mathbf{g}_k
:= \big( g_{k,1},\ g_{k,2},\ \ldots,\ g_{k,M} \big)^{\top}.
\]
To incorporate the local geometry of the quantum state space, we use the Fubini--Study metric, which serves as the Riemannian metric tensor on the space of pure quantum states. The entries of the metric are given as
\[ [\mathrm{F}_k]_{(i,j)}:= \text{Cov}(\Upsilon_{i}(\boldsymbol{\beta}),\Upsilon_{j}(\boldsymbol{\beta}))_{|\Psi^{(k-1)}\rangle},\]
where $$\Upsilon_{i}(\boldsymbol{\beta}):= -\mathrm{i}(\partial_{i} U_{\mathcal{P}}^{\dagger}(\boldsymbol{\beta}) )U_{\mathcal{P}}(\boldsymbol{\beta}) = \mathrm{i} U_{\mathcal{P}}^\dagger(\boldsymbol{\beta})(\partial_{i}U_{\mathcal{P}}(\boldsymbol{\beta})),$$ is an observable and the covariance between two observables $A$ and $B$ acting on the state $|\Psi\rangle$ is defined as 
\begin{align*}
\mathrm{Cov}(A,B)_{|\Psi\rangle}
\!:=& \tfrac{1}{2}\langle \Psi |\{A,B\}|\Psi\rangle
    \rangle\!- \!\langle \Psi|A|\Psi\rangle
      \langle \Psi|B|\Psi\rangle .
\end{align*}
Evaluating the metric
at $\boldsymbol{\beta} = \boldsymbol{0}$ gives
\begin{equation}\label{eqn:outer_fisher}
    [\mathrm{F}_k]_{(i,j)}
    = \mathrm{Cov}(O_i,O_j)_{|\Psi^{(k-1)}\rangle}.
\end{equation}
Given the gradient vector $\mathbf{g}_k$ and the
metric tensor $\mathrm{F}_k$, the natural gradient at the $k$-th iteration is given as
\begin{equation}\label{eqn:natGrad}
    \tilde{\mathbf{g}}_k
    := \mathrm{F}_k^{-1}\,\mathbf{g}_k .
\end{equation}
The $j$-th entry of the natural gradient $\tilde{g}_{k,j}:=(\tilde{\mathbf{g}}_k)_j$ quantifies the steepest-descent direction provided by the operator $O_j$. In other words, it measures how effectively $O_j$ lowers the energy after augmenting it to the current ansatz.

\vspace{3pt}
\noindent\textit{4. Geometric Operator Selection.} If the norm of the natural gradient with respect to the quadratic norm $\|\cdot\|_{\mathrm{F}_k}$ is less than a predetermined threshold~$\varepsilon$, i.e.,
\[
\|\tilde{\mathbf{g}}_k\|_{\mathrm{F}_k} := \sqrt{\mathbf{g}_k^\top {\mathrm{F}_k}^{-1} \mathbf{g}_k} < \varepsilon,
\] then Geo-ADAPT has converged, and the algorithm terminates.
Otherwise, select the
operator with the largest absolute natural gradient,
\[
j_k = \arg\max_{j\in\{1,\ldots,M\}} |\tilde{{g}}_{k,j}|.
\]
If $j_k = j_{k-1}$, i.e., the selected operator $O_{j_k}$ coincides with the immediately preceding 
operator $O_{j_{k-1}}$, then adding it would not add any new information to the ansatz. 
In this case, we skip adding $O_{j_k}$ and proceed to Step 5. Note that we are choosing operators with replacement from the pool.

\noindent\textit{Remark.}
If we take $\mathrm{F}_k = I$, then the geometric-operator selection weighting is removed, and the natural gradient
direction reduces to the standard Euclidean gradient. In this case, Geo‑ADAPT‑VQE becomes
identical to the ADAPT‑VQE.

\noindent Step 5: \textit{Parameter Optimization.}
Once the operator $O_{j_k}$ is selected, it is appended to the current ansatz to form the
$k$-th iteration extended ansatz, given as
\begin{equation}
|\Psi(\boldsymbol{\tilde{\theta}}^{(k-1)})\rangle
    = e^{-\,\mathrm{i}\,\beta O_{j_k}}\,|\Psi(\boldsymbol{\theta}^{(k-1)})\rangle,
\end{equation}
where $\beta$ is the new parameter associated with the selected operator $O_{j_k}$ and $\boldsymbol{\tilde{\theta}}^{(k-1)}:= (\boldsymbol{{\theta}}^{(k-1)},\beta)$ is the extended parameter vector.  
For the inner optimization subroutine at the 
$k$-th outer iteration, the parameter vector is initialized as
\begin{equation}
\boldsymbol{\tilde{\theta}}^{(k,0)} := (\boldsymbol{\theta}^{(k-1)}, 0),
\end{equation}
so that the newly introduced parameter starts at zero while the previously optimized parameters remain unchanged.
These parameters are then jointly optimized using the QNGD update rule
\begin{equation}
\boldsymbol{\tilde{\theta}}^{(k,l+1)}
    = \boldsymbol{\tilde{\theta}}^{(k,l)} - \eta\,[{\mathsf{F}(\boldsymbol{\tilde{\theta}}^{(k,l)})}]^{-1}\,\nabla E(\boldsymbol{\tilde{\theta}}^{(k,l)}),
\end{equation}
where $\eta$ is the learning rate, ${\mathsf F}$ is the Fubini-Study metric, and $$E(\boldsymbol{\tilde{\theta}}^{(k,l)}):= \langle\Psi(\boldsymbol{\tilde{\theta}}^{(k,l)})|H|\Psi(\boldsymbol{\tilde{\theta}}^{(k,l)})\rangle$$ is the energy, both evaluated with respect to the variational state$|\Psi(\boldsymbol{\tilde{\theta}}^{(k,l)})\rangle$. The QNGD updates are applied for a fixed number of optimization steps
$\kappa$. At the end of this optimization stage, we obtain the updated parameter set
$\boldsymbol{\theta}^{(k)}:= \boldsymbol{\tilde{\theta}}^{(k,\kappa)}$.
The algorithm then returns to Step~4 to evaluate the next operator to be added.

In Step 5, we employ QNGD since the operator selection is guided by geometric considerations, and the subsequent parameter optimization should follow the same geometry. Thus, QNGD is used for the inner optimization subroutine, as its updates align with the chosen geometric direction and reinforce the benefits of geometry-aware operator selection.

\begin{algorithm}[ht]
\caption{Geo-ADAPT-VQE}
\label{alg:geo-adapt-vqe}
\DontPrintSemicolon
\setstretch{1.25}
\LinesNumbered
\SetKwFunction{FOpt}{QNGD}

\KwIn{Hamiltonian $\hat{H}$, initial state $\ket{\Psi_{\mathrm{HF}}}$, operator pool $\mathcal{P}\! =\!\{O_j\}_{j=1}^M$,   tolerance $\varepsilon > 0,$ maximum outer iterations $K,$ and maximum inner iterations $\kappa$}
\KwOut{Updated Ansatz $|\Psi(\boldsymbol{\theta}^{(K)})\rangle$}

\tcc{Initialization}
Set $\boldsymbol{\theta}^{(0)} \leftarrow \emptyset$, and $\ket{\Psi^{(0)}} \leftarrow \ket{\Psi_{\mathrm{HF}}}$\;

\For{$\normalfont k = 1$ to $K$}
{
\For{$j = 1$ \KwTo $M$}{
    \tcc{Compute gradient}
    ${g}_{k,j} \leftarrow -\mathrm{i}\langle \Psi^{(k-1)} | [\,\hat{H}, O_j\,] | \Psi^{(k-1)} \rangle$\;

    \tcc{Compute information metric}
    \For{$i = 1$ \KwTo $M$}{
        $[\mathrm{F}_k]_{(i,j)} \leftarrow 
       \mathrm{Cov}(O_i,O_j)_{|\Psi^{(k-1)}\rangle}$\;
    }
}
    $\tilde{\mathbf{g}}_k \leftarrow \mathrm{F}_k^{-1}{\mathbf{g}}_k$\;
 \medskip
  \tcc{Stopping criterion}
    \If{$\|\tilde{\mathbf{g}}_{k}\|_{\mathrm{F}_k}< \varepsilon$}{
        \textbf{break}\;
    }
 \medskip  
    \tcc{Geometric operator selection}
    Select index $j_k := \arg\max_{j} |\tilde{g}_{k,j}|$\;
\If{${j_k} \neq {j_{k-1}}$}{
$|\Psi(\boldsymbol{\theta}^{(k-1)}, \beta)\rangle \leftarrow e^{-\mathrm{i}\beta O_{j_k}} |\Psi(\boldsymbol{\theta}^{(k-1)})\rangle$\;
}
\medskip
            \tcc{Inner QNGD optimization}
    Initialize $\boldsymbol{\tilde{\theta}}^{(k,0)} \leftarrow (\boldsymbol{\theta}^{(k-1)}, 0)$\;
    \For{$\ell = 0$ \KwTo $\kappa - 1$}{
        Compute $\mathsf{F}(\boldsymbol{\tilde{\theta}}^{(k,l)})$ and $\nabla E(\boldsymbol{\tilde{\theta}}^{(k,l)})$\;
        $\boldsymbol{\tilde{\theta}}^{(k,\ell+1)} \leftarrow \boldsymbol{\tilde{\theta}}^{(k,\ell)} - \eta\,[\mathsf{F}(\boldsymbol{\tilde{\theta}}^{(k,l)})]^{-1} \nabla E(\boldsymbol{\tilde{\theta}}^{(k,l)})$\;
    }
    $\boldsymbol{\theta}^{(k)} \leftarrow \boldsymbol{\tilde{\theta}}^{(k,\kappa)}$\;
    Updated ansatz $|\Psi(\boldsymbol{\theta}^{(k)})\rangle$
}
\KwRet{$ |\Psi(\boldsymbol{\theta}^{(K)})\rangle$}
\end{algorithm}

\subsection{Convergence Analysis}
\label{subsec:convergence}
In this section, we present the convergence analysis of the Geo-ADAPT algorithm. We begin with the following assumptions on the energy $\phi^{(k)}$ associated with the trail-extended ansatz and the pool information metric $\mathrm{F}_k$.  

\vspace{5pt}
\noindent\textbf{A1 (L-smooth trial-extended energy).}
For each outer iteration $k$, the trial-extended energy $\phi^{(k)}$ is coordinate-wse $L$-Lipschitz continuous, i.e., for each coordinate $j$ and $\beta\in\mathbb{R}$, we have
\[
\phi^{(k)}(\beta\mathbf{e}_j)
\;\le\;
\phi^{(k)}(\boldsymbol{0})
+
\beta\, g_{k,j}
+
\frac{L}{2}\,\beta^2,\]
where the constant $L>0$ and 
$\mathbf{e}_j$ is a vector with a one in position $j$ and zero in all other positions.

\vspace{2pt}
\noindent\textbf{A2 (Bounded Pool Information Metric).}
For the metric $\mathrm{F}_k$ defined in \eqref{eqn:outer_fisher}, that there exists constants 
$0 < \mu < \lambda < \infty$
such that
\[
\mu \mathrm{I} 
\;\preceq\; 
\mathrm{F}_k 
\;\preceq\;
\lambda \mathrm{I} ,
\qquad \text{for all}\, k.
\]
\textbf{A3 (Diagonal Dominance).}
For each outer iteration $k$, 
the element of the pool information metric \eqref{eqn:outer_fisher} satisfies
\[
\sum_{t\neq j_k}|[\mathrm{F}_k]_{(j_k,t)}|
\le
(1-\rho)[\mathrm{F}_k]_{(j_k,j_k)}, 
\]
where $j_k = \arg\max_{j\in\{1,\ldots,M\}} |\tilde{{g}}_{k,j}|.$

\vspace{2pt}
\noindent\textbf{A4: (Approximate Optimization Subroutine)}
Let $\bar{E}^{(k)} := \inf_{\boldsymbol{\tilde{\theta}}^{(k)}}\,{E}(\boldsymbol{\tilde{\theta}}^{(k)})$ 
denote the optimal energy achievable by the extended ansatz at the $k$-th outer iteration.
Assume that the inner optimization subroutine returns $\boldsymbol{\theta}^{(k)}$ at the $k$-th outer iteration
satisfying
\[
E(\boldsymbol{\theta}^{(k)}) \;\le\; \bar{E}^{(k)} +\delta_k,
\]
where $\delta_k \ge 0 \eqand 
\sum_{k=1}^{\infty}\delta_k < \infty.$
The constant $\delta_k$ models the fact that the inner optimization subroutine might not fully optimize the extended ansatz to its optimal energy at each outer iteration. For instance, $\delta_k$ may arise if the inner optimization is stopped
after reaching a predetermined computational budget, or if it converges to a local
minima. When the subroutine reaches the
optimal energy $\bar{E}^{(k)}$, then $\delta_k = 0$.

\vspace{5pt}
With Assumptions A1 through A4 in place, we now discuss the convergence analysis for Geo-ADAPT. A key ingredient is the \textit{descent property for a single outer iteration}. By descent, we mean that once an operator $O_{j_k}$ is selected at the $k$-th outer iteration using the operator selection rule, the inner optimization routine must produce an updated parameter $\beta$, such that for every outer iteration $k$, we have
$\phi^{(k)}(\beta\mathbf{e}_{j_k}) <  \phi^{(k)}(\boldsymbol{0}) =  E^{(k-1)}.$
In the following lemma, we demonstrate that the geometric operator-selection rule, along with the QNGD inner optimization subroutine of Geo-ADAPT, guarantees this property. A detailed proof can be
found in Appendix \ref{app:proof:lem:descent}.
\begin{lemma}\label{lemma:descent}
Let $j_k = \arg\max_{j\in\{1,\ldots,M\}} |\tilde{g}_{k,j}|$
be the index selected at the $k$-th outer iteration. For a learning rate
$\eta = \mu / L$, consider the inner optimization update
\[
{\beta} = - \eta\,\frac{{g}_{k,j_k}}{[\mathrm{F}_k]_{(j_k,j_k)}}.
\]
Then, for each outer iteration $k$, the following inequality holds
\[
\phi^{(k)}({\beta \mathbf{e}_{j_k}})
\;\le\;
E^{(k-1)}
-\rho^2\frac{\mu}{2L}\,
\tilde{g}_{k,j_k}^2.
\]
\end{lemma}
Hence, whenever \(|\tilde g_{k,j_k}|\neq 0\), the selected coordinate produces strict descent:
$\phi^{(k)}(\beta \mathbf e_{j_k}) < E^{(k-1)}.$
We now incorporate the effect of the inner optimization
subroutine. In the following lemma, we summarize the energy descent bound following the inner optimization subroutine at the $k$-th outer iteration. A proof is provided in Appendix \ref{app:proof:lem:inner}.
  
\begin{lemma}
\label{lem:inner}
Let
$\boldsymbol{\theta}^{(k)}$ be the $k$-length parameter vector returned by the inner optimization subroutine.
Then, the following inequality holds
\begin{equation}
E(\boldsymbol{\theta}^{(k)})
\;\le\;
E^{(k-1)}
-
\rho^2\frac{\mu}{2L}
\tilde{g}_{k,j_k}^2
+
\delta_k, 
\label{eq:approx-descent}
\end{equation}
where $\delta_k\geq0$ such that $\sum_{k}\delta_k < \infty$.
\end{lemma}

Next, in the following theorem, we use this lemma to argue that the sequence $\{E^{(k)}\}$ converges to some finite value. As a result, the maximum natural-gradient score $\max_j|\tilde{g}_{k,j}|$ vanishes in the limit $k\rightarrow\infty$. Thus, asymptotically, it establishes that no operator in the pool contributes non-zero to the energy descent, and hence the algorithm converges to a stationary point in the pool. Finally, we strengthen this asymptotic convergence result by using
an assumption (QGI inequality \cite{sohail2025quantum}), stated below, to derive an exponential convergence rate for
$E^{(k)}$, and thus asymptotically converges to the global minimum achieved by any possible ansatz constructed using the operator pool $\mathcal{P}$.

\vspace{5pt}
\noindent\textbf{A5: (QGI inequality for pool information metric)}
There exists a constant $\mu_0>0$ such that, the following inequality holds for all $k$,
\[
\frac{1}{2}\;{\mathbf{g}_k^\top {\mathrm{F}_k}^{-1} \mathbf{g}_k}
\;\ge\;
\mu_0\big(\phi^{(k)}(0) - E^\star\big),
\]
where $E^*$ is the global minimum energy achievable by any ansatz constructed using the operator pool $\mathcal{P}$.
\begin{theorem}\label{thm:asymptotic_conv}
    Let $\{E^{(k)}\}_{k\geq 0}$ be the sequence of energies generated by Geo-ADAPT-VQE
Then, under the assumptions A1-A4, the following statements hold true:
\begin{enumerate}
    \item (Asymptotic Converges) The sequence $\{E^{(k)}\}_{k}$ converges to a finite value $E_{\infty}\geq E^*$, i.e., 
    $$E^{(k)} \longrightarrow E_{\infty}\quad \text{as}\quad k\rightarrow\infty,$$ 
    where $E^*$ is the global minimum energy achievable by any ansatz constructed using $\mathcal{P}$.
\item The pool natural gradient vanishes asymptotically: 
    $$\|\tilde{\mathbf{g}}_k\|
\;\longrightarrow\; 0
\quad\text{as } k\to\infty,$$
The limiting point of the algorithm is a pool-stationary point. 
\end{enumerate}
Moreover, under the outer QGI condition (Assumption~A5), if $\rho\leq \sqrt{\frac{2\lambda M L}{\mu \mu_0}}$, then the energy sequence satisfies the following exponential convergence bound with a vanishing residual:
\[
E^{(k)} - E^\star
\;\le\;
\Big(1 - \frac{\rho^2\mu \mu_0}{2\lambda M L}\Big)^{k}\,(E^{(0)} - E^\star)
+ e_k,
\]
where \(e_k \to 0\) as \(k \to \infty\) and
\(E^{(0)} := \langle\Psi_{\mathrm{HF}}|H|\Psi_{\mathrm{HF}}\rangle.\)
\end{theorem}
We provide a detailed proof in Appendix \ref{app:proof:thm:asymptotic_conv}.

\begin{figure*}[h]
    \centering
    \includegraphics[scale=0.3]{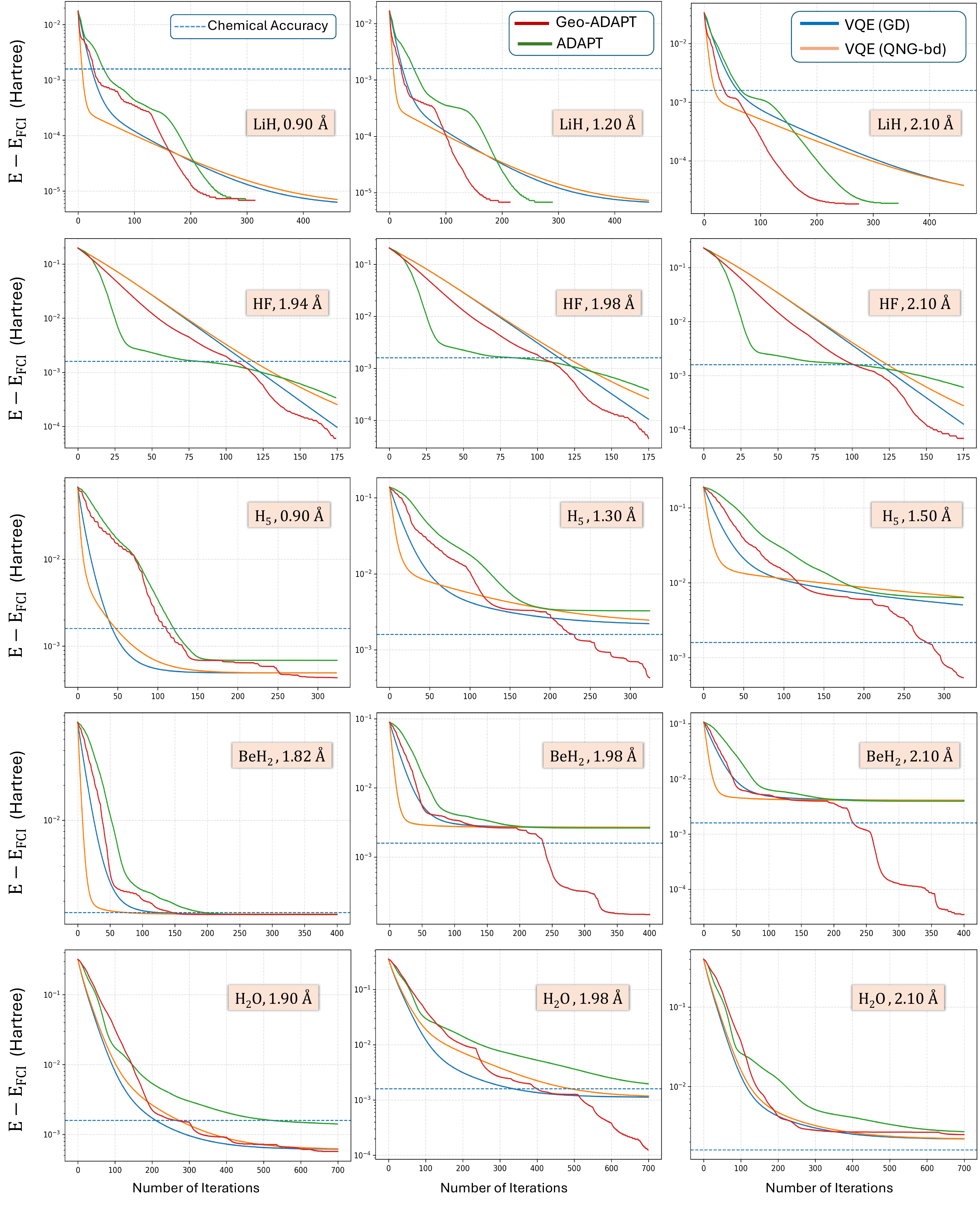}
    \caption{Energy error vs.\ number of iterations comparing GD, QNG-bd, ADAPT-VQE, and Geo-ADAPT-VQE.
    Geo-ADAPT consistently requires fewer steps to achieve chemical accuracy and exhibits improved convergence stability, especially for longer bond lengths.}
    \label{fig:error_vs_steps}
\end{figure*}

\begin{figure*}[h]
    \centering
    \includegraphics[scale=0.3]{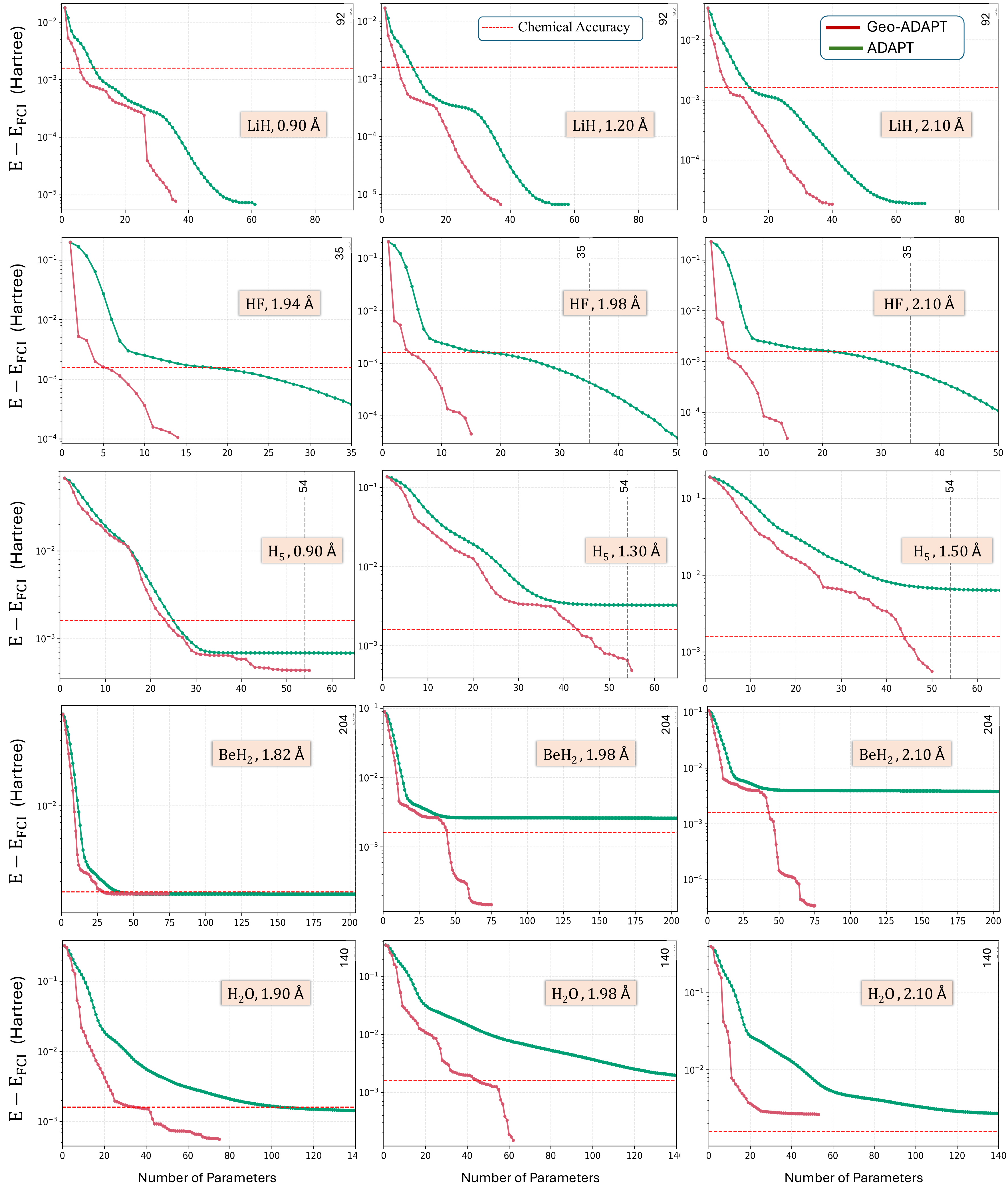}
    \caption{Energy error vs.\ number of ansatz parameters for \ce{LiH}, HF, \ce{H5}, \ce{BeH2}, and \ce{H2O} across several bond lengths.
    Geo-ADAPT reaches chemical accuracy with substantially fewer operators than ADAPT-VQE.}
    \label{fig:error_vs_params}
\end{figure*}

\subsection{Numerical Results}
\label{subsec:numerical}

To demonstrate the utility of Geo-ADAPT, we numerically evaluate its performance on \ce{LiH}, \ce{HF}, \ce{H5}, \ce{BeH2}, and \ce{H2O} across multiple bond lengths, using molecular Hamiltonians from the PennyLane quantum chemistry dataset~\cite{Utkarsh2023Chemistry}. We benchmark Geo-ADAPT against ADAPT-VQE and standard VQE optimized using both GD and QNGD. For ADAPT-VQE, we employ standard gradient descent as the parameter optimizer. For VQE (QNGD), we use a block-diagonal approximation of the Fubini--Study metric (QNG-bd), since computing the full metric for the UCCSD ansatz at each iteration is computationally prohibitive. The operator pool comprises generators of all single- and double-excitation operators, and the VQE baseline employs the UCCSD ansatz. All energies are reported in Hartrees, and learning rates for all methods are selected via grid search over the range $10^{-1}$ to $10^{-4}$. We fix the number of inner optimization steps at $\kappa = 5$ across all experiments, unless otherwise stated. Thus, for example, a maximum of $100$ total iterations corresponds to $20$ outer adaptive steps. For the Geo-ADAPT inner optimization subroutine, we compute the full metric at each iteration because the ansatz grows gradually in size, and the dimensionality of the parameter space remains moderate, making full metric evaluation computationally feasible.

In Fig.~\ref{fig:error_vs_steps}, we report the energy error (relative to the reference FCI energy) as a function of the total number of optimization iterations. For \ce{LiH}, both VQE (GD) and VQE (QNG-bd) initially converge faster than ADAPT-VQE and Geo-ADAPT. However, as the number of iterations increases, their convergence slows down significantly. In contrast, the adaptive algorithms exhibit slower initial progress but ultimately converge more effectively. In particular, Geo-ADAPT converges faster than ADAPT-VQE and achieves the same energy error level of approximately $10^{-5}$ with roughly $20$--$25\%$ fewer iterations. 
For \ce{HF}, ADAPT-VQE converges rapidly at the beginning, reaching close to chemical accuracy ($\sim 3 \times 10^{-3}$) within about $30$ iterations. However, its performance subsequently saturates, requiring approximately $145$ additional iterations to reach an energy error of roughly $3 \times 10^{-4}$. In contrast, Geo-ADAPT progresses more gradually initially but attains a comparable error $(\sim 3 \times 10^{-4})$ with approximately $50$ fewer iterations. In this case, VQE (GD) and VQE (QNG-bd) show similar behavior and exhibit significantly slower convergence overall.

For \ce{H5}, \ce{BeH2}, and \ce{H2O}, the convergence behavior of ADAPT-VQE and VQE is qualitatively similar across bond lengths. For smaller bond lengths, the energy error saturates near chemical accuracy, and for larger bond lengths, it often fails to reach that threshold. In contrast, Geo-ADAPT achieves a comparable initial convergence speed but does not saturate at chemical accuracy. Instead, it continues to reduce the energy error. On average, Geo-ADAPT improves the energy error by up to one order of magnitude. In some cases, such as with \ce{BeH2} at $2.10$~Å, improvements of up to $100$x have been observed, all achieved with the same number of total optimization iterations across all methods.
The observed saturation of fixed-ansatz VQE can be attributed to the increasingly ill-conditioned optimization landscape of high-dimensional parameter spaces. Fixed ansatz, such as UCCSD, often include many excitation operators that contribute only marginally to variational energy reduction. The inclusion of redundant operators increases circuit depth and leads to an over-parameterized optimization landscape, without providing meaningful descent directions. As a result, this does not result in significant improvements in energy reduction \cite{hashimoto2026comprehensive,larocca2023theory,garcia2024effects}. Moreover, such landscapes are known to exhibit barren plateau behavior, characterized by vanishing gradients that significantly slow gradient-based optimization. In addition, approximations such as the block-diagonal quantum natural gradient (QNG-bd) neglect inter-parameter correlations that become important for high-precision refinement, potentially leading to early saturation to suboptimal energy levels.

In ADAPT-VQE, although the ansatz is constructed adaptively, the operator selection rule relies solely on the magnitude of first-order energy gradients. As the ansatz approaches a locally stationary region, gradients across the operator pool can become uniformly small. This results in gradient troughs, which are extended flat regions in the evolution of the energy error \cite{stadelmann2025strategies}. Consequently, the gradient-based selection rule may include operators that provide only marginal improvements, leading to diminishing energy savings and early saturation.
In contrast, Geo-ADAPT selects operators along the natural gradient direction in the space of quantum states, thereby accounting for the intrinsic geometric structure of the state space. This geometry-aware selection mitigates the effects of flat directions, enabling sustained descent and continued energy reduction beyond the saturation level observed in other methods.

In Fig.~\ref{fig:error_vs_params}, we report the energy error as a function of the number of operators in the ansatz, comparing Geo-ADAPT with ADAPT-VQE. The dashed line indicates the total number of excitation operators in the corresponding UCCSD ansatz. 
For \ce{LiH}, Geo-ADAPT achieves an energy error of approximately $10^{-5}$ with up to a $1.5\times$ reduction in the number of operators compared to ADAPT-VQE. For \ce{HF}, Geo-ADAPT requires approximately $3\times$ fewer operators to reach chemical accuracy, and up to $5\times$ fewer operators to achieve an energy error on the order of $10^{-4}$.
For \ce{H5}, \ce{BeH2}, and \ce{H2O}, particularly at larger bond lengths, ADAPT-VQE fails to reach chemical accuracy even when the number of operators approaches that of the number of excitation operators in the full UCCSD ansatz for each molecule. In comparison, Geo-ADAPT achieves chemical accuracy using fewer than $50$ operators, resulting in approximately a $75\%$ and $30\%$ reduction relative to UCCSD for BeH2 and H2O, respectively.

In Table \ref{tab:VQE_vs_GeoADAPT}, we summarize the reduction in \emph{effective ansatz complexity} (EAC) achieved by Geo-ADAPT relative to VQE (GD) and VQE (QNGD). We define the EAC metric as the total variational overhead required to reach chemical accuracy, given by the product of the number of ansatz operators and the number of iterations. For GD and QNGD, the ansatz is fixed to the number of excitation operators in the UCCSD ansatz. The results indicate that Geo-ADAPT consistently achieves chemical accuracy with significantly lower EAC. For instance, Geo-ADAPT provides an average improvement of over $80\%$ compared to VQE (QNG-bd), an $84\%$ improvement over VQE (GD), and up to $82\%$ improvement over ADAPT-VQE.
\begin{table}[h!]
\centering
\begin{tabular}{|l|c|c|c|c|}
\hline
\multirow{2}{*}{\small{\textbf{Molecule}}} 
& \multirow{2}{*}{\footnotesize{\textbf{Qubits}}} 
& \multicolumn{3}{c|}{\footnotesize{\textbf{Gain over}}} \\
\cline{3-5}
&  & \footnotesize{\textbf{GD}} & \footnotesize{\textbf{QNG-bd}} &  \footnotesize{\textbf{ADAPT-VQE}}\\
\hline
\noalign{\vskip 2pt}
\small{\ce{LiH} $(2.1$\AA$)$}     & 12  & $95.31\%$ & $85.18\%$ & $72.65\%$\\
\small{HF $(1.98$\AA$)$}          & 12  & $89.66\%$ & $90.18\%$ & $65.33\%$\\
\small{\ce{BeH2} $(1.82$\AA$)$}   & 14  & $82.73\%$ & $70.98\%$ & $39.35\%$\\
\small{\ce{H2O} $(1.9$\AA$)$}     & 14  & $69.47\%$ & $76.43\%$ & $82.05\%$\\
\hline
\end{tabular}
\caption{Comparison of the \emph{effective ansatz complexity (EAC)} required 
to reach chemical accuracy for each molecule using three optimization schemes: 
GD, QNG-bd, and Geo-ADAPT-VQE. The reported percentages correspond to the 
relative improvement of Geo-ADAPT over GD and QNG-bd, respectively.}
\label{tab:VQE_vs_GeoADAPT}
\end{table}

\begin{figure*}[t]
    \centering
    \includegraphics[scale=0.35]{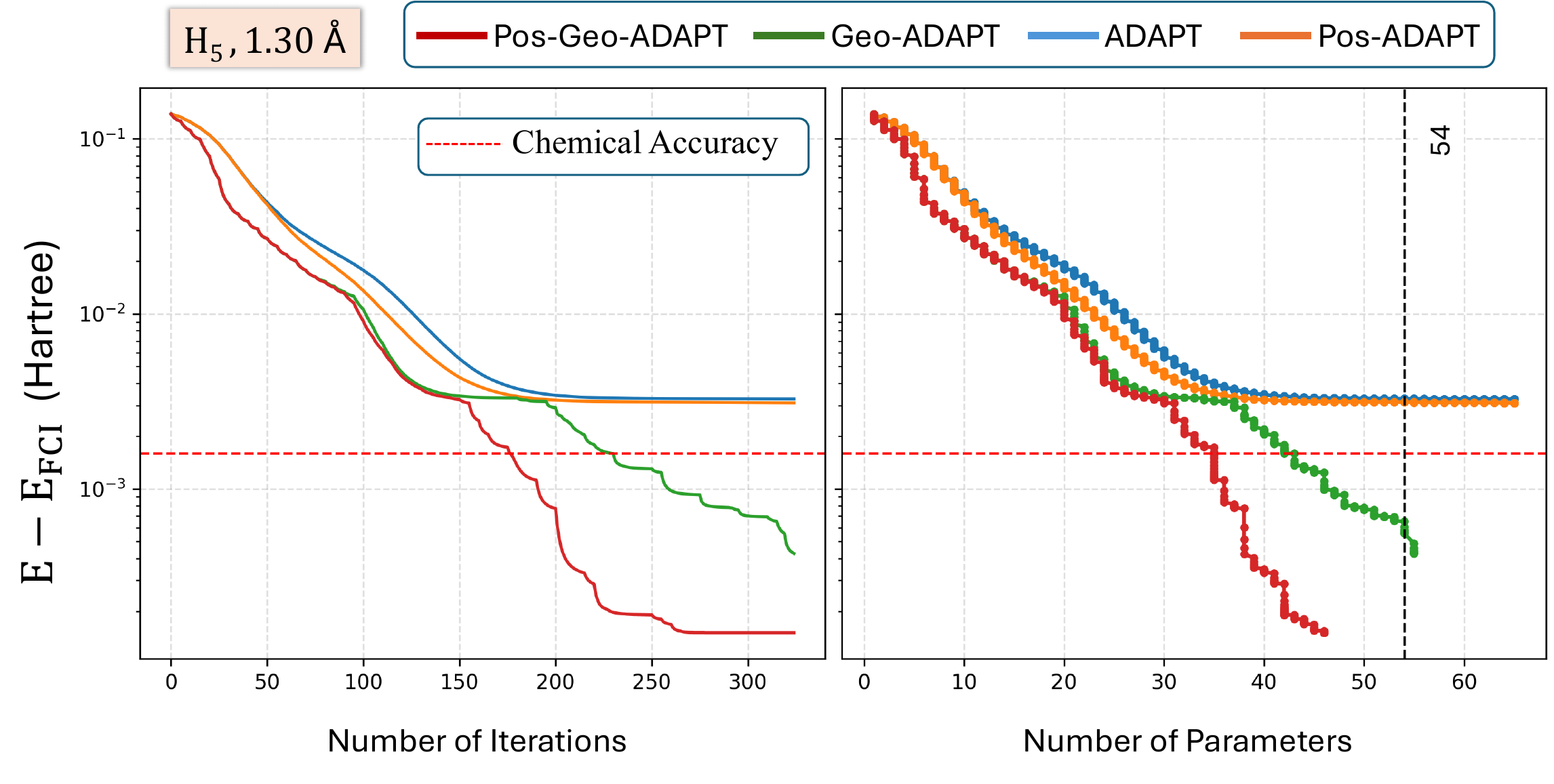}
    \caption{Performance comparison for \ce{H5} at $1.3\,\text{\AA}$. Pos-Geo-ADAPT achieves faster convergence and lower energy error than ADAPT-VQE and Geo-ADAPT, whereas Pos-ADAPT exhibits performance similar to ADAPT-VQE. The vertical dashed line denotes the number of parameters in the UCCSD ansatz.}\label{fig:posGeo}
\end{figure*}
\noindent \textbf{Pos-Geo-ADAPT.}
Building on Geo-ADAPT, we further enhance the ansatz construction by optimizing the placement of operators within the ansatz. In addition to selecting the operator based on the natural gradient magnitude, Pos-Geo-ADAPT optimizes the insertion position of the selected operator within the ansatz. In ADAPT and Geo-ADAPT, the selected operators are appended at the end of the ansatz. However, due to the non-commutative nature of excitation operators, the position of an operator can substantially affect its effective action on the state. Pos-Geo-ADAPT exploits this flexibility by optimizing the insertion position at each step. The detailed procedure is provided in Appendix~\ref{app:posgeo}.

In Fig.~\ref{fig:posGeo}, we report the energy error as a function of both the total number of optimization iterations and the number of ansatz operators, comparing Pos-Geo-ADAPT with Geo-ADAPT. For \ce{H5} at a bond length of $1.3$~\AA, Pos-Geo-ADAPT achieves approximately $3\times$ lower energy error than Geo-ADAPT after $325$ iterations. Moreover, Pos-Geo-ADAPT results in $36.6\%$ reduction in EAC compared to Geo-ADAPT.
We also compare Pos-Geo-ADAPT with the positional counterpart of ADAPT-VQE, known as Pos-ADAPT. Unlike the geometric case, positional refinement does not show significant performance improvements for ADAPT. Both ADAPT and Pos-ADAPT demonstrate similar convergence behavior, reaching saturation before achieving chemical accuracy.

\section{Discussion}
We introduced Geo-ADAPT-VQE, a geometry-aware adaptive variational quantum eigensolver that integrates operator selection with the intrinsic geometry of quantum states. We further proposed Pos-Geo-ADAPT, which refines both operator choice and insertion position to enhance expressive efficiency.
Overall, the numerical experiments indicate that Geo-ADAPT achieves lower energy error with fewer parameters and fewer optimization steps across a wide range of molecular systems. The improvements observed in both energy-versus-iteration and energy-versus-ansatz-size analyses highlight the effectiveness of integrating geometric operator selection with natural-gradient based optimization. In particular, Geo-ADAPT consistently avoids the early saturation behavior observed in gradient-only methods and continues reducing the energy error beyond chemical accuracy.

We further investigated the role of the inner optimization subroutine. In Appendix~\ref{app:innerK}, we provide the effect of the number of inner optimization updates $\kappa$. Our results indicate the trade-off between converges speed and number of ansatz parameters. Smaller values of $\kappa$ provides faster convergence. However, number of ansatz parameters required to achieve energy error increases. Consequently, the computational cost, particularly the evaluation of the information metric becomes more expensive. These observations suggest the selection of an intermediate $\kappa$ in which convergence speed and ansatz size are balanced. Identifying the optimal choice of $\kappa$ that jointly optimizes convergence speed, ansatz size, and computational cost remains an interesting direction for future investigation.

Next, in Appendix \ref{app:ADAPTQNG}, we provide the effect of the inner optimization methods. In particular, we
compared trade-off among four variants: Pos-Geo-ADAPT with QNGD, Pos-Geo-ADAPT with GD,
Pos-ADAPT with QNGD, and Pos-ADAPT with GD. The comparison is performed for
\ce{H5} at $1.3$~\AA\ and \ce{LiH} at $2.1$~\AA. We observe that
when QNGD is used within Pos-ADAPT, the optimization initially follows
a similar descent trajectory as Pos-Geo-ADAPT but subsequently develops
strong oscillations and fails to converge reliably. Moreover, Pos-Geo-ADAPT with GD shows convergence behavior similar to Pos-ADAPT with GD. These results suggest that achieving stable and faster convergence requires aligning the operator selection rule with the parameter update direction according to the same underlying information geometry.

An important question concerns the role of the second-order optimization method, such as the Newton-Raphson method, which utilizes the Hessian of the
loss function. The update step of the Newton method replaces the gradient by multiplying the gradient by the
inverse of the Hessian matrix. However, this does not address saddle points satisfactorily, and instead, saddle
points become attractive under Newton dynamics \cite{dauphin2014identifying}. Since
variational circuits typically exhibit many saddle points,
this limits the practical effectiveness of Hessian-based methods. Another important direction concerns the choice of classical optimization
methods used in the inner optimization subroutine.
Beyond natural-gradient methods, a variety of optimizers have been
explored in the VQE literature, including adaptive gradient-based
methods such as AdaGrad and Adam, as well as gradient-free strategies. While recent
works~\cite{stokes2020quantum,tao2023laws,yamamoto2019natural,wierichs2020avoiding} suggest that natural-gradient approaches
can offer improved stability and convergence speed for VQE, a systematic comparison of these optimizers within adaptive ansatz construction frameworks remains an interesting direction
for future work. Performing such a comparison fairly would require
careful consideration of learning-rate schedules, adaptive step-size
strategies, and metric estimation costs, as well as a deeper theoretical
analysis of convergence behavior.

Beyond comparing inner optimization methods, an important direction is extending Geo-ADAPT to noisy quantum settings. In mixed-state scenarios, the appropriate geometric object is the Bures metric. However, computing the full metric may be sample-intensive. One promising approach is the use of ensemble-based quantum Fisher information metrics~\cite{sohail2025quantum}, which can reduce sample complexity while preserving geometric structure. Moreover, investigating the convergence rate and sample complexity of Geo-ADAPT under realistic noise models remains a topic that requires further exploration.

Beyond quantum chemistry, geometry-aware adaptive circuit construction
could be used in areas such as quantum machine learning and
variational quantum sensing~\cite{khan2025quantum,schuld2020circuit,kaubruegger2023variational,du2022quantum,roth2025autoqml,koike2022autoqml,martyniuk2024quantum}.
In these settings, parameterized quantum circuits are widely used as
trainable models, and their performance depends critically on the
expressivity of the ansatz and the structure of the optimization
landscape. Geometry-aware adaptive strategies could therefore provide a
principled mechanism for constructing expressive yet trainable circuits
tailored to specific learning or sensing objectives. Establishing
theoretical guarantees for such approaches, including convergence
analysis and sample-complexity bounds, remains an important direction for future
research.

\bibliography{references.bib}

\onecolumngrid
\newpage
\appendix
\section*{Appendices}
\section{Effect of Inner QNGD Iterations}
\label{app:innerK}
In Geo-ADAPT-VQE, once a new operator is selected by the geometric
selection rule, the parameters are optimized using
QNGD. This inner optimization is
performed for a fixed number of iterations $\kappa$. The choice of
$\kappa$ directly affects the convergence behavior of the algorithm. Fig.~\ref{fig:error_vs_params} illustrates the effect of varying the
number of inner QNGD iterations $\kappa$ on the convergence of the
algorithm for the $\ce{H5}$ system at bond length $1.30\,\text{\AA}$.
inner quantum natural gradient descent (QNGD) routine. This inner
optimization is performed for a fixed number of iterations $\kappa$.
The choice of $\kappa$ influences the balance between the quality of
the parameter optimization and the overall computational cost.

Fig.~\ref{fig:innerK} shows the effect of varying the number
of inner QNGD iterations $\kappa$ for $\ce{H5}$ at a
bond length of $1.30\,\text{\AA}$. The top panels correspond to
ADAPT-VQE with different values of $\kappa$, while the bottom panels
show the corresponding results for Geo-ADAPT-VQE. In each case, the
energy error is plotted as a function of the number of parameters and
the number of total optimization iterations.
For ADAPT-VQE (top panels), change in $\kappa$ has no effect on
the achieved energy error. The convergence behavior remains largely
unchanged across different values of $\kappa$. However, it
affects the convergence rate. Smaller values of
$\kappa$ tend to reach the plateau region more quickly in terms of the
number of iterations, but this comes at the expense of
introducing more parameters into the ansatz before convergence. From the
figure, moderate values such as $\kappa=2$ or $\kappa=3$ appear to
provide the best balance, achieving convergence with fewer parameters.

In Geo-ADAPT-VQE (bottom panels) shows a stronger dependence
on $\kappa$. Decreasing the number of inner iterations allows to reach
lower energy errors with fewer iterations. 
However, this improvement comes with a trade-off. Smaller values of
$\kappa$ lead to deeper circuits with more parameters
before convergence, which increases both the circuit depth and
the computational cost associated with evaluating the metric
tensor and performing repeated QNGD updates. Therefore, $\kappa$
controls a trade-off between faster convergence and the
overall computational overhead of the algorithm. From the
figure, $\kappa=4$ seems optimal considering energy error, number iterations, and ansatz parameters.
\begin{figure*}[h]
    \centering
    \includegraphics[scale=0.36]{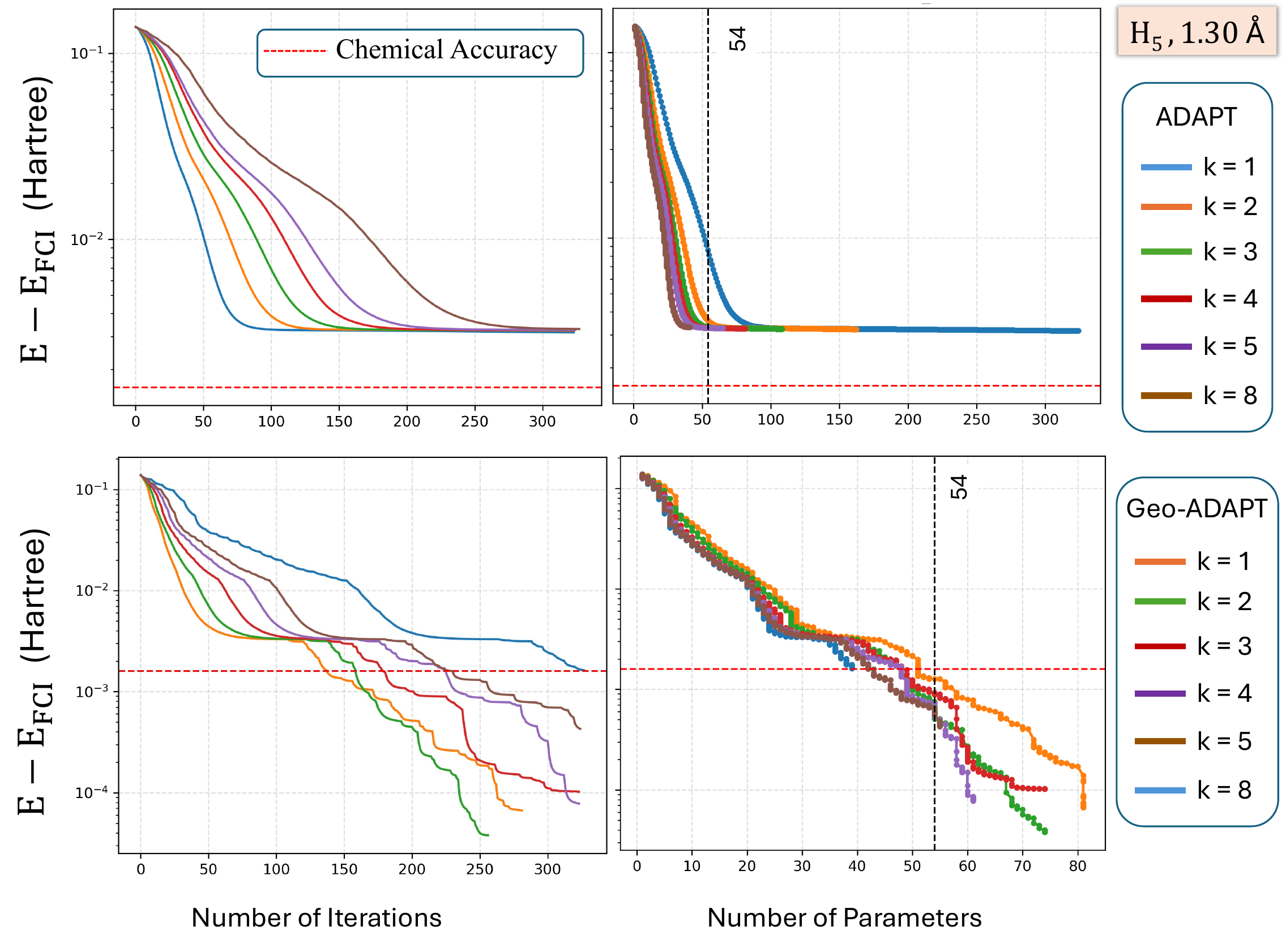}
    \caption{Effect of the number of inner QNGD iterations $\kappa$ on the convergence behavior for the $\ce{H5}$ system at $1.30\,\text{\AA}$. The top panels show ADAPT-VQE and the bottom panels show Geo-ADAPT-VQE. Energy error is plotted as a function of the number of parameters (left) and the number of total optimization iterations (right). Increasing $\kappa$ improves convergence efficiency in Geo-ADAPT-VQE but introduces a trade-off between faster outer-loop convergence and the number of parameters added to the ansatz.}
    \label{fig:innerK}
\end{figure*}

\section{Pos-Geo-ADAPT}
\label{app:posgeo}
Geo-ADAPT-VQE selects, at each outer iteration, the operator from the pool
that gives the steepest energy descent with respect to the information geometry
of the current variational state, and then appends that operator to the end of
the ansatz. However, in a noncommutative variational circuit, the action of a
generator depends not only on which operator is chosen, but also on where it is
inserted in the circuit. With this motivation, we introduce Pos-Geo-ADAPT.
At each outer iteration, the algorithm jointly optimizes over both the operator
index and its insertion position. 

Suppose that after the $(k-1)$-th outer iteration, the ansatz is
\begin{equation}
U^{(k-1)}(\boldsymbol{\theta}^{(k-1)})
=
\prod_{t=1}^{{k-1}} e^{-\mathrm{i}\theta_t^{(k-1)} O_{j_t}},
\end{equation}
The corresponding variational state is
$\ket{\Psi^{(k-1)}}
=
U^{(k-1)}(\boldsymbol{\theta}^{(k-1)})\ket{\Psi_{\mathrm{HF}}}.$
For a candidate insertion position
$p\in\{0,1,\dots,{(k-1)}\}$, define the prefix and suffix unitary by
\begin{equation}
U^{(k-1)}_{\le p}
:=
\prod_{t=1}^{p} e^{-\mathrm{i}\theta_t^{(k-1)} O_{j_t}}
\quad \text{and} \quad
U^{(k-1)}_{>p}
:=
\prod_{t=p+1}^{k-1} e^{-\mathrm{i}\theta_t^{(k-1)} O_{j_t}},
\end{equation}
so that
$U^{(k-1)}(\boldsymbol{\theta}^{(k-1)})
=
U^{(k-1)}_{>p}\,U^{(k-1)}_{\le p}.$
Here, $p=0$ corresponds to inserting before the first gate, while
$p={(k-1)}$ corresponds to appending at the end.
If the pool operator $O_i\in\mathcal P$ is inserted at position $p$, then the
extended ansatz takes the form
\begin{equation}
U_{\mathrm{trial}}^{(k-1)}(\boldsymbol{\beta},p)
=
U^{(k-1)}_{>p}\bigg(\prod_{j=1}^M e^{-\mathrm{i}\beta_j O_j}\bigg)U^{(k-1)}_{\le p},
\end{equation}
and the corresponding trial state is
$\ket{\tilde\Psi^{(k)}(\boldsymbol{\beta},p)}
=
U_{\mathrm{trial}}^{(k-1)}(\boldsymbol{\beta},p)\ket{\Psi_{\mathrm{HF}}}$ (see Fig.\ref{fig:pos-split-wide}) and the energy is given as
\begin{equation}
\phi^{(k)}(\boldsymbol{\beta},p)
=
\bra{\tilde\Psi^{(k)}(\boldsymbol{\beta},p)}
\hat H
\ket{\tilde\Psi^{(k)}(\boldsymbol{\beta},p)}.
\end{equation}
\begin{figure}[h]
\centering
\begin{tikzpicture}[>=stealth,thick]
\node[draw,rounded corners,minimum width=1.9cm,minimum height=0.9cm,align=center] (Upre)
{$U^{(k-1)}_{\le p}$};
\node[draw,rounded corners,minimum width=2.5cm,minimum height=1.0cm,align=center,right=1.4cm of Upre] (V)
{$\prod_{i=1}^{M} e^{-\mathrm{i}\beta_i O_i}$};
\node[draw,rounded corners,minimum width=1.9cm,minimum height=0.9cm,align=center,right=1.4cm of V] (Upost)
{$U^{(k-1)}_{>p}$};
\draw[->] ([xshift=-1.9cm]Upre.west) -- (Upre.west)
node[midway,above] {$\ket{\Psi_{\mathrm{HF}}}$};
\draw[->] (Upre.east) -- (V.west);
\draw[->] (V.east) -- (Upost.west);
\draw[->] (Upost.east) -- ([xshift=2.5cm]Upost.east)
node[midway,above] {$\ket{\tilde\Psi^{(k)}(\boldsymbol{\beta},p)}$};
\draw[dashed] ($(Upre.east)!0.5!(V.west)+(0,0.5)$) -- ($(Upre.east)!0.5!(V.west)+(0,-0.5)$);
\node at ($(Upre.east)!0.5!(V.west)+(0,-0.85)$) {$p$};
\end{tikzpicture}
\caption{Schematic of Pos-Geo-ADAPT at a fixed insertion position $p$.}
\label{fig:pos-split-wide}
\end{figure}
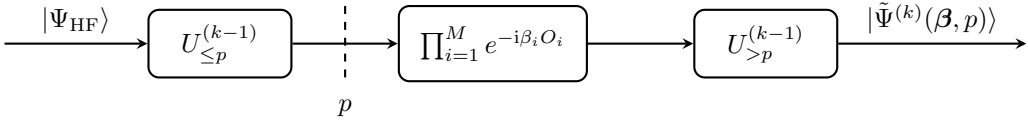

The first-order energy sensitivity for inserting $O_i$ at position $p$ is
captured by
\begin{equation}
g_{k,(i,p)}
:=
\frac{d}{d\boldsymbol{\beta}}\phi^{(k)}(\boldsymbol{\beta},p)\Big|_{\boldsymbol{\beta}=0}
=
-\mathrm{i}\bra{\Psi^{(k-1)}}
[\hat H,\tilde O_{i,p}^{(k)}]
\ket{\Psi^{(k-1)}},
\quad \text{where }
\tilde O_{i,p}^{(k)}
:=
\big(U^{(k-1)}_{>p}\big)^\dagger O_i\,U^{(k-1)}_{>p}.
\end{equation}
For each fixed position $p$, we collect the gradients into the vector
$\mathbf g_{k,p}
=
\big(
g_{k,(1,p)},g_{k,(2,p)},\dots,g_{k,(M,p)}
\big)^{\intercal}.$
To incorporate the geometry of the quantum state space, we define the
position-dependent pool information metric as
\begin{equation}
[\mathrm F_{k,p}]_{(i,j)}
=
\mathrm{Cov}\!\left(
\tilde O_{i,p}^{(k)},
\tilde O_{j,p}^{(k)}
\right)_{\ket{\Psi^{(k-1)}}}.
\end{equation}
The associated natural gradient is then
$\tilde{\mathbf g}_{k,p}
=
\mathrm F_{k,p}^{-1}\mathbf g_{k,p},$
and its $i$-th component $\tilde g_{k,(i,p)}$ quantifies the geometric
descent contribution of inserting operator $O_i$ at position $p$.
The joint operator-position selection rule is therefore
\begin{equation}
(i_k,p_k)
=
\arg\max_{\substack{i\in\{1,\dots,M\}\\ p\in\{0,\dots,(k-1)\}}}
\left|\tilde g_{k,(i,p)}\right|.
\end{equation}
As in Geo-ADAPT-VQE, we terminate when the natural-gradient becomes
sufficiently small. Since the search is now position-dependent, the
stopping criterion is based on $p_k$:
\begin{equation}
\|\tilde{\mathbf g}_{k,p_k}\|_{\mathrm F_{k,p_k}}
<\varepsilon.
\end{equation}
Once the optimal pair $(i_k,p_k)$ is selected, the unitary
$e^{-\mathrm{i}\beta O_{i_k}}$ is inserted at position $p_k$ in the ansatz, and all
variational parameters are re-optimized using QNGD, consistent with the
geometry-aware outer selection rule. The completes the description of Pos-Geo-ADAPT, and summarized in Algorithm ~\ref{alg:pos-geo-adapt-vqe}.
\begin{algorithm}[h]
\caption{Pos-Geo-ADAPT-VQE}
\label{alg:pos-geo-adapt-vqe}
\DontPrintSemicolon
\setstretch{1.25}
\LinesNumbered

\KwIn{Hamiltonian $\hat{H}$, reference state $\ket{\Psi_{\mathrm{HF}}}$, operator pool $\mathcal P=\{O_i\}_{i=1}^{M}$, tolerance $\varepsilon>0$, maximum outer iterations $K$, and maximum inner QNGD iterations $\kappa$}
\KwOut{Final ansatz $\ket{\Psi(\boldsymbol{\theta}^{(K)})}$}

\tcc{Initialization}
Set $\boldsymbol{\theta}^{(0)}\leftarrow \emptyset$, $U^{(0)}\leftarrow I$, and $\ket{\Psi^{(0)}}\leftarrow \ket{\Psi_{\mathrm{HF}}}$\;

\For{$k=1$ \KwTo $K$}{

    \tcc{Compute position-dependent gradients and metrics}
    \For{$p=0$ \KwTo $(k-1)$}{

        \tcc{Define prefix and suffix unitaries}
        $U_{\le p}^{(k-1)} \leftarrow \prod_{t=1}^{p} e^{-i\theta_t^{(k-1)}O_{j_t}}$\;
        $U_{>p}^{(k-1)} \leftarrow \prod_{t=p+1}^{(k-1)} e^{-i\theta_t^{(k-1)}O_{j_t}}$\;

        \For{$i=1$ \KwTo $M$}{
            $\tilde O_{i,p}^{(k)} \leftarrow \big(U_{>p}^{(k-1)}\big)^\dagger O_i\,U_{>p}^{(k-1)}$\;

            \tcc{Position-wise energy gradient}
            $g_{k,(i,p)} \leftarrow -i\bra{\Psi^{(k-1)}}[\hat H,\tilde O_{i,p}^{(k)}]\ket{\Psi^{(k-1)}}$\;
        
\tcc{Position-dependent information metric}
            \For{$j=1$ \KwTo $M$}{
                $[\mathrm F_{k,p}]_{(i,j)}
                \leftarrow
                \mathrm{Cov}\!\big(
                \tilde O_{i,p}^{(k)},
                \tilde O_{j,p}^{(k)}
                \big)_{\ket{\Psi^{(k-1)}}}$\;
            }
        }

        \tcc{Natural gradient at position $p$}
        $\tilde{\mathbf g}_{k,p}\leftarrow \mathrm F_{k,p}^{-1}\mathbf g_{k,p}$\;
    }
    Select
    \(
    (i_k,p_k)
    =
    \arg\max_{\substack{i\in\{1,\dots,M\}\\ p\in\{0,\dots,(k-1)\}}}
    \big|\tilde g_{k,(i,p)}\big|
    \) \;

    \tcc{Stopping criterion}
    \If{$\|\tilde{\mathbf g}_{k,p_k}\|_{\mathrm F_{k,p_k}}<\varepsilon$}{
        \textbf{break}\;
    }

    \tcc{Insert selected operator at position $p_k$}
    \(
    \ket{\Psi(\boldsymbol{\theta}^{(k-1)},\beta)}
    \leftarrow
     U_{>p_k}^{(k-1)}\,e^{-i\beta O_{i_k}}\,U_{\le p_k}^{(k-1)}(\beta)\ket{\Psi_{\mathrm{HF}}}.
    \)
    
\tcc{Inner QNGD optimization}
   Initialize $\boldsymbol{\tilde{\theta}}^{(k,0)} \leftarrow (\boldsymbol{\theta}^{(k-1)}, 0)$\;

    \For{$\ell=0$ \KwTo $\kappa-1$}{
        Compute $\mathsf F(\tilde{\boldsymbol{\theta}}^{(k,\ell)})$ and $\nabla E(\tilde{\boldsymbol{\theta}}^{(k,\ell)})$\;
       $\tilde{\boldsymbol{\theta}}^{(k,\ell+1)}
        \leftarrow
        \tilde{\boldsymbol{\theta}}^{(k,\ell)}
        -
        \eta\,
        \big[\mathsf F(\tilde{\boldsymbol{\theta}}^{(k,\ell)})\big]^{-1}
        \nabla E(\tilde{\boldsymbol{\theta}}^{(k,\ell)})$
    }

    $\boldsymbol{\theta}^{(k)}\leftarrow \tilde{\boldsymbol{\theta}}^{(k,\kappa)}$\;
    Updated ansatz 
    $\ket{\Psi^{(k)}}(\boldsymbol{\theta}^{(k)})$\;
}
\KwRet{$\ket{\Psi(\boldsymbol{\theta}^{(K)})}$}
\end{algorithm}

\section{Effect of the Inner Optimization Method}
\label{app:ADAPTQNG}
We investigate the role of the inner optimization method in the
overall performance of the adaptive algorithms. In particular, we
compare four variants: Pos-Geo-ADAPT with QNGD, Pos-Geo-ADAPT with
GD, Pos-ADAPT with QNGD, and Pos-ADAPT
with GD. Similar to our Pos-Geo-ADAPT, Pos-ADAPT denotes the position-aware variant of ADAPT-VQE, in which the
operator is selected using the standard gradient-based rule but
can be inserted at any position in the current ansatz rather than
being appended only at the end of the circuit. Fig.~\ref{fig:adaptQNG} shows the convergence behavior for
$\ce{H5}$ at bondlength $1.3\,\text{\AA}$ and \ce{LiH} at bondlength $2.1\,\text{\AA}$.

For both molecules, Pos-Geo-ADAPT with QNGD exhibits stable and
consistent convergence, reaching chemical accuracy without
oscillations. Pos-ADAPT with GD also shows stable convergence.
However, its convergence is slower for \ce{LiH} compared to
Pos-Geo-ADAPT, and for \ce{H5} it saturates at a higher energy error
than Pos-Geo-ADAPT. Next, when the geometric operator selection is combined
with GD inner optimization, the convergence behavior becomes similar to that of
Pos-ADAPT with GD, suggesting that geometry-aware operator selection
alone is insufficient if the parameter updates do not follow the same
information geometry in the inner optimization subroutine.

Replacing GD with QNGD in Pos-ADAPT provides performance similar to Geo-ADAPT with QNGD in the beginning, but eventually exhibits strong oscillations
and fails to converge reliably. This behavior is particularly evident
for \ce{LiH}, where the error fluctuates significantly after the initial
descent. These oscillations persist even when the learning rate is
reduced from $\eta=0.1$ to $\eta=0.085$ and $0.05$. Note that for $\ce{H5}$,  decreasing the learning rate to $\eta = 0.05$ does provide improvement in terms of
stability. However, the overall convergence behavior then begins to
resemble that of standard ADAPT-VQE, with slower convergence and higher energy error. 
These observations suggest that the operator selection rule and the
parameter update direction must be aligned with the same underlying
information geometry to achieve stable, efficient, and faster convergence,
which motivates the construction of Geo-ADAPT.
\begin{figure*}[h]
    \centering
    \includegraphics[scale=0.405]{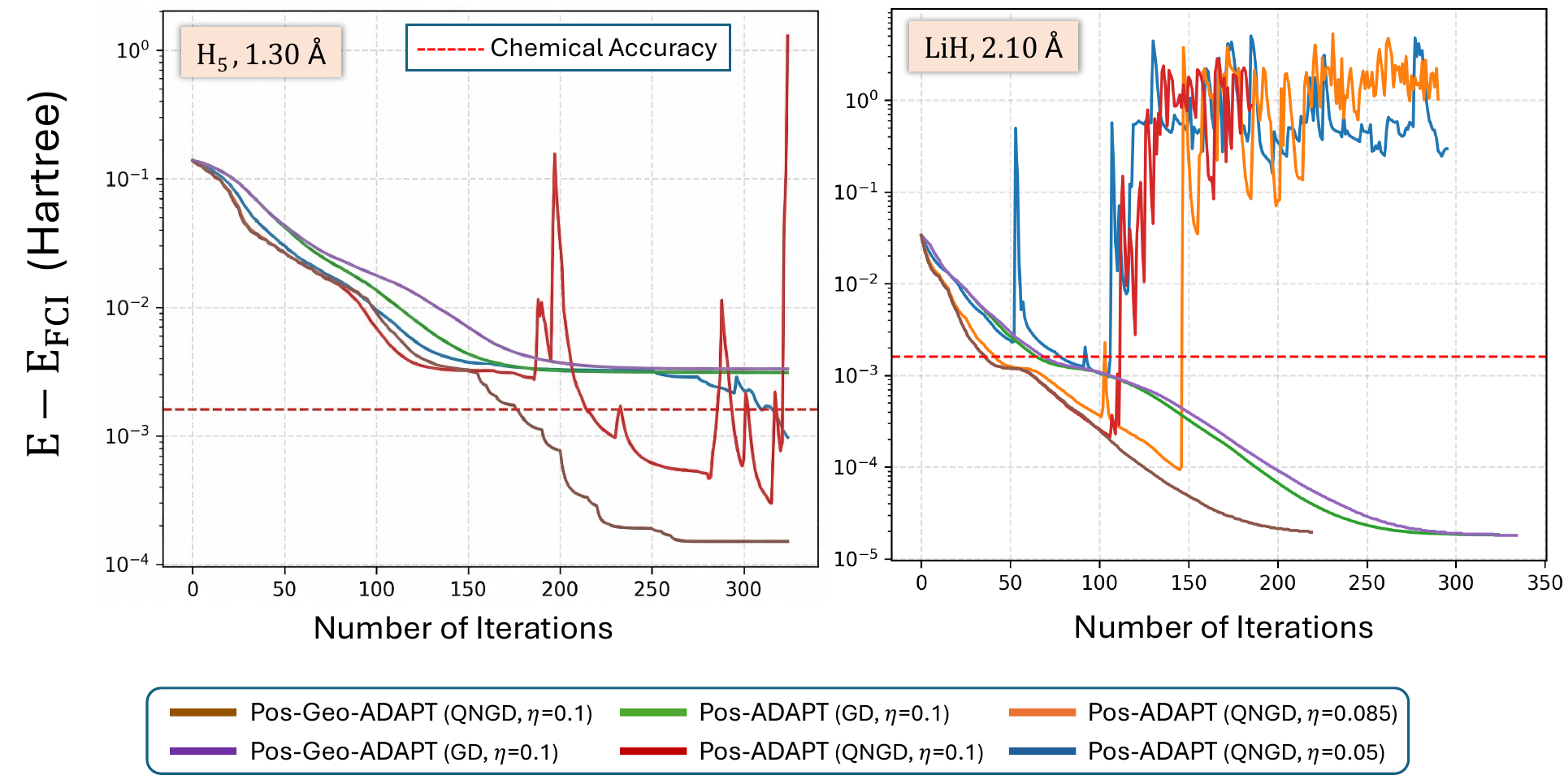}
    \caption{Effect of the inner optimization method on the convergence behavior.
The energy error is plotted versus the number
of total optimization iterations for $\ce{H5}$ at $1.3\,\text{\AA}$ (left)
and \ce{LiH} at $2.1\,\text{\AA}$ (right). We compare Pos-Geo-ADAPT and Pos-ADAPT
using either QNGD or GD in the inner optimization.
A fixed learning rate $\eta=0.1$ is used, while additional runs with smaller
learning rates ($\eta=0.085$ and $0.05$) are also shown for ADAPT-QNGD.}
    \label{fig:adaptQNG}
\end{figure*}
\section{Proof of Lemma \ref{lemma:descent}}
\label{app:proof:lem:descent}

Under Assumption~A1, we first establish a one-step descent bound for the selected coordinate \(j_k\). By the coordinate-wise \(L\)-smoothness of \(\phi^{(k)}\) and the identity \(\phi^{(k)}(\boldsymbol{0}) = E^{(k-1)}\), we have
\begin{align}
\phi^{(k)}(\beta \mathbf{e}_{j_k})
\overset{}{\leq}
E^{(k-1)} + \beta\, g_{k,j_k} + \frac{L}{2}\beta^2 
&\overset{a}{\leq}
E^{(k-1)}
- \eta\,\frac{g_{k,j_k}^2}{[\mathrm{F}_k]_{(j_k,j_k)}}
+\frac{L}{2}\eta^2\left(\frac{g_{k,j_k}}{[\mathrm{F}_k]_{(j_k,j_k)}}\right)^2 \notag\\
&=
E^{(k-1)}
-
\left(
\eta [\mathrm{F}_k]_{(j_k,j_k)}-\frac{L}{2}\eta^2
\right)
\left(
\frac{g_{k,j_k}}{[\mathrm{F}_k]_{(j_k,j_k)}}
\right)^2,
\label{eq:one_step_descent_prelim}
\end{align}
where \((a)\) is obtained by substituting
$\beta=-\eta({g_{k,j_k}}/{[\mathrm{F}_k]_{(j_k,j_k)}}).$
Next, recalling from \eqref{eqn:natGrad} that
$\mathbf g_k = \mathrm{F}_k \tilde{\mathbf g}_k,$
we obtain
\[
g_{k,j_k}
=
\sum_{t=1}^M [\mathrm{F}_k]_{(j_k,t)} \tilde g_{k,t}
=
[\mathrm{F}_k]_{(j_k,j_k)} \tilde g_{k,j_k}
+
\sum_{t\neq j_k}[\mathrm{F}_k]_{(j_k,t)} \tilde g_{k,t} =
[\mathrm{F}_k]_{(j_k,j_k)} \tilde g_{k,j_k}\,\rho_{k,j_k},
\]
where
\begin{equation}
\rho_{k,j_k}
:=
1+
\sum_{t\neq j_k}
\frac{[\mathrm{F}_k]_{(j_k,t)}\,\tilde g_{k,t}}
     {[\mathrm{F}_k]_{(j_k,j_k)}\,\tilde g_{k,j_k}}.
\label{eq:rho_def}
\end{equation}
After substituting this identity into \eqref{eq:one_step_descent_prelim}, we get
\begin{equation}
\phi^{(k)}(\beta \mathbf e_{j_k})
\le
E^{(k-1)}
-
\rho_{k,j_k}^2
\left(
\eta [\mathrm{F}_k]_{(j_k,j_k)}-\frac{L}{2}\eta^2
\right)
\tilde g_{k,j_k}^2.
\label{eq:one_step_descent_rho}
\end{equation}
It remains to lower bound \(\rho_{k,j_k}\). From \eqref{eq:rho_def}, the triangle inequality gives
\begin{align}
\rho_{k,j_k}
&\ge
1-
\sum_{t\neq j_k}
\left|
\frac{[\mathrm{F}_k]_{(j_k,t)}\,\tilde g_{k,t}}
     {[\mathrm{F}_k]_{(j_k,j_k)}\,\tilde g_{k,j_k}}
\right| =
1-
\sum_{t\neq j_k}
\frac{|[\mathrm{F}_k]_{(j_k,t)}|}{[\mathrm{F}_k]_{(j_k,j_k)}}
\left|
\frac{\tilde g_{k,t}}{\tilde g_{k,j_k}}
\right| \overset{b}{\ge}
1-
\sum_{t\neq j_k}
\frac{|[\mathrm{F}_k]_{(j_k,t)}|}{[\mathrm{F}_k]_{(j_k,j_k)}} \overset{c}{\ge}
\rho > 0.
\label{eq:rho_lower_bound}
\end{align}
where \((b)\) follows from the definition of \(j_k\), i.e.,
$|\tilde g_{k,t}| \le |\tilde g_{k,j_k}|$ for all $t\neq j_k \in \{1,2,\cdots,M\}$
and $(c)$ uses the Assumption~A3 (diagonal-dominance):
\[\sum_{t\neq j_k}|[\mathrm{F}_k]_{(j_k,t)}|
\le
(1-\rho)[\mathrm{F}_k]_{(j_k,j_k)}\quad \text{for all}\ k.\]
Finally, combining \eqref{eq:one_step_descent_rho} and \eqref{eq:rho_lower_bound}, we get
\begin{equation}
\phi^{(k)}(\beta \mathbf e_{j_k})
\le
E^{(k-1)}
-
\rho^2
\left(
\eta [\mathrm{F}_k]_{(j_k,j_k)}-\frac{L}{2}\eta^2
\right)
\tilde g_{k,j_k}^2.
\label{eq:one_step_descent_final}
\end{equation}
Next, observe that any learning rate satisfying
$0<\eta<\tfrac{2\mu}{L}$
ensures that
$(\eta [\mathrm{F}_k]_{(j_k,j_k)}-\frac{L}{2}\eta^2) >0
$ for all $k$.
In particular, choosing \(\eta=\mu/L\) in \eqref{eq:one_step_descent_final} gives
\[
\phi^{(k)}(\beta \mathbf e_{j_k})
\le
E^{(k-1)}
-
\rho^2\frac{\mu}{2L}\,\tilde g_{k,j_k}^2.
\]
This establishes the desired descent property and completes the proof of Lemma~\ref{lemma:descent}.
\hfill\(\square\)

\section{Proof of Lemma \ref{lem:inner}}
\label{app:proof:lem:inner}
Recall from Lemma~\ref{lemma:descent} that for 
$j_k = \arg\max_j |\tilde{g}_{k,j}|$, we have
\[
\phi^{(k)}(\beta\mathbf{e}_{j_k})
\;\leq\;
E^{(k-1)}
-
\rho^2\frac{\mu}{2L}\tilde{g}_{k,j_k}^2,
\]
where $\beta = -\eta({g_{k,j_k}}/{[\mathrm{F}_k]_{(j_k,j_k)}})$ and 
$\eta= \mu/2L$. Next, consider the extended ansatz at the $k$-th iteration obtained by 
appending the operator $O_{j_k}$ with parameter $\beta$ to the 
current ansatz $|\Psi^{(k-1)}\rangle$. This construction implies that using the inner optimization subroutine yields the following inequalities 
\[
\bar{E}^{(k)} = \inf_{\boldsymbol{\tilde{\theta}}^{(k)}} E(\boldsymbol{\tilde{\theta}}^{(k)}) 
\leq \inf_{\beta'} \phi^{(k)}(\beta'\mathbf{e}_{j_k})
\leq  \phi^{(k)}(\beta \mathbf{e}_{j_k}).
\]
Therefore, using Lemma~\ref{lemma:descent} yields
\begin{equation}
\bar{E}^{(k)} 
\;\leq\;
E^{(k-1)} - \rho^2\frac{\mu}{2L}\tilde{g}_{k,j_k}^2 .
\label{eq:trial-energy-bound}
\end{equation}
By Assumption~A4, the inner optimization subroutine returns 
$\boldsymbol{\theta}^{(k)}$ satisfying
$E(\boldsymbol{\theta}^{(k)}) 
\;\leq\;
\bar{E}^{(k)} +\delta_k,$
for some $\delta_k \ge 0$ with $\sum_{k=1}^{\infty}\delta_k < \infty$.
Combining this with \eqref{eq:trial-energy-bound} gives
\[
E(\boldsymbol{\theta}^{(k)}) 
\;\leq\;
E^{(k-1)} - \rho^2\frac{\mu}{2L}\tilde{g}_{k,j_k}^2 +\delta_k.
\]
This completes the proof of Lemma~\ref{lem:inner}.\hfill$\square$

\section{Proof of Theorem \ref{thm:asymptotic_conv}}
\label{app:proof:thm:asymptotic_conv}
We begin by establishing the asymptotic convergence of the energy sequence. 
From Lemma~\ref{lem:inner}, we have
\[
E(\boldsymbol{\theta}^{(k)}) 
\;\leq\;
E^{(k-1)} - \rho^2\frac{\mu}{2L}\tilde{g}_{k,j_k}^2 +\delta_k.
\]
Let the tail of the perturbation series be defined as 
$R_k := \sum_{t = k+1}^{\infty} \delta_t.$
Since $\sum_{t=1}^{\infty} \delta_t < \infty$, it follows from \cite[Theorem 3.23]{rudin1976principles} that the sequence $(\delta_t)$ converges to zero. Consequently, applying \cite[Theorem 3.1.9]{bartle2000introduction}, we obtain
$\lim_{k \to \infty} R_k = 0.$ 
Define the auxiliary energy sequence
$\tilde{E}^{(k)} := E^{(k)} + R_k.$
We now show that $\tilde{E}^{(k)}$ is monotonically non-increasing. 
Starting from the inequality above, we obtain
\begin{align*}
    E^{(k)} + R_k 
    &\leq 
    E^{(k-1)} - \rho^2\frac{\mu}{2L}\,\tilde{g}_{k,j_k}^2 
    + (\delta_k + R_k)
    \implies 
    \tilde{E}^{(k)} 
    \leq 
    \tilde{E}^{(k-1)} - \rho^2\frac{\mu}{2L}\,\tilde{g}_{k,j_k}^2 .
\end{align*}
Therefore,
\(\tilde{E}^{(k)} \leq \tilde{E}^{(k-1)},
\) for all $k>0$.
Furthermore, since $E^{(k)} \geq E^*$ and $R_k \geq 0$, we have 
$\tilde{E}^{(k)} \geq E^*$.  
This implies $\{\tilde{E}^{(k)}\}$ is bounded below and monotone non-increasing, therefore, using the monotone convergence theorem \cite[Theorem 3.14]{rudin1976principles},  $\tilde{E}^{(k)}$
converges to a finite limit $E_{\infty} \geq E^*$.  
Next, recall $\tilde{E}^{(k)} = E^{(k)} + R_k$ and $R_k \to 0$, therefore from \cite[Theorem 3.4]{rudin1976principles}, it follows that 
\[
\lim_{k\rightarrow\infty}E^{(k)} = \lim_{k\rightarrow\infty}\tilde{E}^{(k)} - \lim_{k\rightarrow\infty}R_k = E_{\infty}.
\]
Thus, $E^{(k)}$ also converges to the finite limit $E_{\infty} \geq E^*$.
Next, we show that the pool natural gradient vanishes asymptotically.  
Summing the bound from Lemma~\ref{lem:inner} from $k = 1$ to $T$ yields
\[
E^{(T)} 
\leq 
E^{(0)} - \rho^2\frac{\mu}{2L} \sum_{k=1}^T \tilde{g}_{k,j_k}^2
+ \sum_{k=1}^T\delta_k.
\]
Since $\lim_{k\rightarrow\infty}E^{(k)}= E_{\infty}$ and $\sum_k\delta_k < \infty$. Taking
$T \to \infty$ gives
\[
\sum_{k=1}^{\infty} \tilde{g}_{k,j_k}^2
\;\leq\;
\frac{2L}{\rho^2 \mu}\bigg((E^{(0)} - E_\infty) + \sum_{k=1}^{\infty}\delta_k\bigg)
\;<\; \infty.
\]
Therefore, 
$|\tilde{g}_{k,j_k}| \longrightarrow 0.$
By definition of $j_k$, this implies
$\max_j |\tilde{g}_{k,j}| \longrightarrow 0,$
and hence 
$\tilde{\mathbf{g}}_k \longrightarrow \boldsymbol{0}.$
Under Assumption~A3, each $\mathrm{F}_k$ is positive definite, so 
\[
\mathrm{F}_k^{-1}\mathbf{g}_k \to 0 
\;\;\Longleftrightarrow\;\;
\mathbf{g}_k \to \boldsymbol{0}.
\]
Thus, all directional derivatives vanish in the limit, and the algorithm reaches a state in which no operator in the pool can further decrease the energy, i.e., the limit point of the algorithm is a pool-stationary point.

\medskip
\noindent
We now assume the QGI inequality holds, and Assumption~A4 ensures 
$\mathrm{F}_k \preceq \lambda I$.  
Then,
\begin{align*}
    \tilde{g}_{k,j_k}^2 
= \|\tilde{\mathbf{g}}_k\|_{\infty}^2
&\geq \frac{1}{M}\,\|\tilde{\mathbf{g}}_k\|_2^2\geq \frac{1}{\lambda M}\, \|\tilde{\mathbf{g}}_k\|_{\mathrm{F}_k}^2\geq \frac{2\mu_0}{\lambda M}\,(E^{(k-1)} - E^*).
\end{align*}
Since $\tilde{\mathbf{g}}_k \to 0$ and $\mathrm{F}_k \succ 0$, the left-hand side 
must vanish, therefore 
\[
\lim_{k\rightarrow\infty}(E^{(k)} - E^{*}) = 0 \implies \lim_{k\rightarrow\infty}E^{(k)} =E^{*}. 
\]
To obtain the exponential convergence rate, we combine the above inequality with 
Lemma~\ref{lem:inner}, giving
\[
E^{(k)} 
\;\leq\;
E^{(k-1)} 
- \frac{\rho^2\mu \mu_0}{2\lambda M L}\,(E^{(k-1)} - E^{*}) 
+\delta_k.
\]
Finally, subtracting $E^*$ from both sides and applying the above inequality recursively yields
\[
(E^{(k)} - E^{*}) 
\leq 
\left(1 - \frac{\rho^2\mu \mu_0}{2\lambda M L}\right)^{k} (E^{(0)} - E^{*})
+ e_k,\quad 
\text{ where }
e_k := 
\sum_{t=1}^{k} 
\left(1 - \frac{\rho^2\mu \mu_0}{2\lambda M L}\right)^{(k-t)} \delta_t.
\]
If $\rho\leq \sqrt{\frac{2\lambda M L}{\mu \mu_0}}$, then $e_k\rightarrow0
\text{ as } k \to \infty$ (see Lemma~\ref{app:lem:perturbed_linear} in Appendix \ref{app:perturbed_linear}). As a consequence, we observe that $\lim_{k\rightarrow\infty}(E^{(k)} - E^{*}) = 0$. 
This completes the proof of Theorem~\ref{thm:asymptotic_conv}. 
\hfill$\square$

\section{Convergence of a perturbed linear recursion}\label{app:perturbed_linear}
In this section, we state and prove a technical result on the convergence 
of perturbed linear recursions. This auxiliary lemma is used in the proof 
of Theorem~\ref{thm:asymptotic_conv} to handle the summable error terms 
arising from the approximate inner optimization subroutine.
\begin{lemma}\label{app:lem:perturbed_linear}
Let $\{\Delta_k\}_{k\ge 0}$ and $\{\delta_k\}_{k\ge 1}$ be real sequences with
$\Delta_k \ge 0$ and $\delta_k \ge 0$ for all $k$. Suppose that
\begin{equation*}
    \Delta_{k+1} \;\le\; \rho\,\Delta_k + \delta_{k+1},
    \qquad k = 0,1,2,\dots,
    \label{eq:recursion}
\end{equation*}
for some constant $\rho \in (0,1)$ and $\sum_{k=1}^{\infty}\delta_k < \infty.$
Then, for all $k \ge 0$,
\begin{equation}\label{eq:Delta-explicit}
    \Delta_k
\le
\rho^k \Delta_0
+
R_k,
\end{equation}
where $R_k := \sum_{m=1}^{k} \rho^{k-m}\delta_m$. Moreover, $R_k \to 0$ as $k \to \infty$, and therefore $\Delta_k \to 0$ as $k \to \infty$.
\end{lemma}
\noindent\textit{Proof.}
We prove \eqref{eq:Delta-explicit} by induction on $k$.
For $k=0$, the inequality reads
\[
    \Delta_0 \;\le\; \rho^0 \Delta_0 + \sum_{m=1}^{0} \rho^{0-m}\delta_m
    = \Delta_0,
\]
which is true with equality.
Assume now that \eqref{eq:Delta-explicit} holds for some $k\ge 0$, i.e.,
$ \Delta_k
    \;\le\;
    \rho^k \Delta_0
    +
    \sum_{m=1}^{k} \rho^{k-m} \delta_m.$
Using the recursion $\Delta_{k+1}
    \le
    \rho\,\Delta_k + \delta_{k+1}$, we obtain
\begin{align*}
    \Delta_{k+1}
    &\le
    \rho \biggl( \rho^k \Delta_0 + \sum_{m=1}^{k} \rho^{k-m} \delta_m \biggr)
    + \delta_{k+1} =
    \rho^{k+1} \Delta_0
    +
    \sum_{m=1}^{k} \rho^{(k+1)-m} \delta_m
    + \delta_{k+1} =
    \rho^{k+1} \Delta_0
    +
    \sum_{m=1}^{k+1} \rho^{(k+1)-m} \delta_m.
\end{align*}
This is exactly \eqref{eq:Delta-explicit} with $k$ replaced by $k+1$. By induction, 
\eqref{eq:Delta-explicit} holds for all $k\ge 0$.

Next, we want to show that $R_k \to 0$ as $k\to\infty$.
Fix an arbitrary $\epsilon > 0$. Using the summability 
assumption $\sum_m\delta_m <\infty$, there exists an integer $\tau$ (depending on 
$\epsilon$) such that
$ \sum_{m=\tau+1}^{\infty} \delta_m < \epsilon.$
For each $k\ge 1$, split $R_k$ as
$ R_k = A_k + B_k,$
where
\begin{equation*}
    A_k := \sum_{m=1}^{\tau} \rho^{k-m} \delta_m
    \eqand
    B_k := \sum_{m=\tau+1}^{k} \rho^{k-m} \delta_m.
\end{equation*}
First, consider $A_k$. For each fixed $m\in\{1,\dots,\tau\}$, we have
$\rho^{k-m} \to 0$ as $k\to\infty$, since $0<\rho<1$. Therefore, for each fixed 
$m$, the term $\rho^{k-m} \delta_m \to 0$ as $k\to\infty$. As $A_k$ is a 
finite sum over $m=1,\dots,\tau$, we conclude that
\begin{equation*}
    A_k \longrightarrow 0 \text{ as } k \rightarrow \infty.
\end{equation*}
Next, consider $B_k$ for $k\ge \tau$. Since $\rho^{k-m} \le 1$ for all $m \leq k$,
we have
\begin{equation*}
    0 \le B_k
    =
    \!\!\!\sum_{m=\tau+1}^{k} \rho^{k-m} \delta_m
    \le
    \sum_{m=\tau+1}^{k} \delta_m
    \le
    \sum_{m=\tau+1}^{\infty} \delta_m
    < \epsilon.
\end{equation*}
Combining these two observations, we obtain that there 
exists $K \ge \tau$ such that for all $k\ge K$,
$|A_k| < \epsilon
    $ and $
    0 \le B_k < \epsilon.$
Hence, for all $k\ge K$,
\[
    0 \le R_k = A_k + B_k \le |A_k| + B_k < 2\epsilon.
\]
Since $\epsilon>0$ was arbitrary, this implies by the definition 
of the limit that
$R_k \longrightarrow 0 \text{ as } k \rightarrow \infty.$
Finally, we show the convergence of $\Delta_k$.
From \eqref{eq:Delta-explicit} and the definition of $R_k$, we 
have
\begin{equation}
    0 \le \Delta_k \le \rho^k \Delta_0 + R_k.
    \label{eq:Delta-upper-bound}
\end{equation}
By the previous steps,
$ \rho^k \Delta_0 \rightarrow 0
    \eqand 
    R_k \rightarrow 0 \quad \text{as } k\rightarrow \infty.$
We now use the limsup to formalize the convergence. Taking $\limsup$ on both 
sides of \eqref{eq:Delta-upper-bound} and using the subadditivity of the limsup i.e., $\limsup (x_k + y_k) \le \limsup x_k + \limsup y_k$, we obtain
\begin{align*}
    \limsup_{k\to\infty} \Delta_k
    &\le
    \limsup_{k\to\infty} \big( \rho^k \Delta_0 + R_k \big) \le
    \limsup_{k\to\infty} \rho^k \Delta_0
    +
    \limsup_{k\to\infty} R_k = 0 + 0 = 0.
\end{align*}
Since $\Delta_k \ge 0$ for all $k$, we also have
\[
    0 \le \liminf_{k\to\infty} \Delta_k \le \limsup_{k\to\infty} \Delta_k \le 0.
\]
Therefore, from \cite[Theorem 3.4.12]{bartle2000introduction}, we conclude that $\lim_{k\to\infty} \Delta_k = 0,$
which completes the proof. \hfill$\square$

\end{document}